\documentclass{article}
\usepackage{arxiv}
\usepackage[utf8]{inputenc} % allow utf-8 input
\usepackage[T1]{fontenc}    % use 8-bit T1 fonts
\usepackage{hyperref}       % hyperlinks
\usepackage{url}            % simple URL typesetting
\usepackage{booktabs}       % professional-quality tables
\usepackage{amsfonts}       % blackboard math symbols
\usepackage{nicefrac}       % compact symbols for 1/2, etc.
\usepackage{microtype}      % microtypography
\usepackage{graphicx}
\usepackage{cite}
\usepackage{amsmath,amssymb,amsfonts, amsthm}
\usepackage{algorithmic}
\usepackage{graphicx}
\usepackage{textcomp}
\usepackage{xcolor}
\usepackage{float}
\usepackage{tabularx}
\usepackage{epstopdf}
\usepackage{graphicx}
\usepackage{colortbl}
\usepackage{caption}
\usepackage{subcaption} 
\usepackage{blkarray}
\usepackage{amsmath}
\usepackage{arydshln}
\usepackage{mathtools}
\usepackage{nicematrix} 
\usepackage{physics}
\NiceMatrixOptions{letter-for-dotted-lines=V}
\newtheorem{theorem}{Theorem}
\newtheorem*{remark}{Remark}

\usepackage{tikz}
\usetikzlibrary{automata,arrows,positioning,calc}
\usepackage{bm}
\usepackage{caption}
\usepackage{subcaption}

\begin{document}
	\title{Exact Analytical Model of Age of Information  in Multi-source Status Update Systems with Per-source Queueing}
	\author{
		Ege~Orkun~Gamgam\\
		Electrical and Electronics Engineering Dept.\\
		Bilkent University, 06800\\
		Ankara, Turkey \\
		\texttt{gamgam@ee.bilkent.edu.tr} \\
		\And
		Nail~Akar\\
		Electrical and Electronics Engineering Dept.\\
		Bilkent University, 06800\\
		Ankara, Turkey \\
		\texttt{akar@ee.bilkent.edu.tr} \\
	}
	\maketitle

% Uncomment to remove the date
%\date{}

% Uncomment to override  the `A preprint' in the header
%\renewcommand{\headeright}{Technical Report}
%\renewcommand{\undertitle}{Technical Report}
\renewcommand{\shorttitle}{}
\renewcommand{\headeright}{A Preprint - October 27, 2021}

%%% Add PDF metadata to help others organize their library
%%% Once the PDF is generated, you can check the metadata with
%%% $ pdfinfo template.pdf
%\hypersetup{
%pdftitle={A template for the arxiv style},
%pdfsubject={q-bio.NC, q-bio.QM},
%pdfauthor={David S.~Hippocampus, Elias D.~Striatum},
%pdfkeywords={First keyword, Second keyword, More},
%}

\begin{abstract}
	We consider an information update system consisting of $N$ sources sending status packets at random instances according to a Poisson process to a remote monitor through a single server.
	We assume a heteregeneous server with exponentially distributed service times which is equipped with a waiting room holding the freshest packet from each source referred to as Single Buffer Per-Source Queueing (SBPSQ). 
	The sources are assumed to be equally important, i.e., non-weighted average AoI is used as the information freshness metric, and subsequently two symmetric scheduling policies are studied in this paper, namely First Source First Serve (FSFS) and the Earliest Served First Serve (ESFS) policies, the latter policy being proposed the first time in the current paper to the best of our knowledge.  By employing the theory of Markov Fluid Queues (MFQ), an analytical model is proposed to obtain the exact distribution of the Age of Information (AoI) for each source when the FSFS and ESFS policies are employed at the server.  Subsequently, a benchmark scheduling-free scheme named as Single Buffer with Replacement (SBR) that uses a single one-packet buffer shared by all sources is also studied with a similar but less complex analytical model. We comparatively study the performance of the three schemes through numerical examples and show that the proposed ESFS policy outperforms the other two schemes in terms of the average AoI and the age violation probability averaged across all sources, in a scenario of sources possessing different traffic intensities but sharing a common service time.
\end{abstract}

% keywords can be removed

\section{Introduction}
Timely delivery of the status packets has been gaining utmost importance in Internet of Things (IoT)-enabled applications \cite{IOT_Survey},\cite{IOT_1} where the information freshness of each IoT device at the destination is crucial, especially for the applications requiring real-time control and decision making. A widely studied metric for quantifying the freshness of data is the Age of Information (AoI) which stands for the time elapsed since the reception of the last status packet at the monitor. More formally, the AoI at time $t$ is defined as the random process $\Delta(t)=t-U(t)$ where $U(t)$ denotes the reception time of the last status packet at the monitor. The AoI metric was first proposed in \cite{firstStudy} for a single-source M/M/1 queueing model and since then a surge of studies followed in the context of a wide range of information update systems \cite{survey},\cite{Yates_Survey}. AoI in multi-source models sharing a single or multiple servers have also been recently studied in several works; see for example \cite{intro_Multi_1},\cite{intro_Multi_2},\cite{intro_Multi_3} and the references therein.

In this paper, we consider an information update system which consists of $N$ sources each of which asynchronously samples an independent stochastic process and subsequently sends these samples in the form of status update packets to a single remote monitor (destination) through a server as shown in Fig.~\ref{figure:system}. Information packets from source-$n, n = 1,2,\dots,N$ are generated according to a Poisson process with rate $\lambda_n$ which contains sensed data along with a time stamp. Generated packets are immediately forwarded to the server with a waiting room (queue) which can contain at most one packet (the freshest) from each source. Therefore, a packet waiting in the queue is replaced with a new fresh packet arrival from the same source. This buffer management is called SBPSQ (Single Buffer Per-Source Queueing). The server is responsible for sending the information packets to the monitor through a communication network which introduces a random service time that is exponentially distributed with parameter $\mu_n$ for source-$n$. A new packet arrival immediately starts to receive service if the server is found idle. On the other hand, SBPSQ needs to be accompanied by a scheduling policy since the server is to choose a source packet among the waiting sources upon a service completion.  In this setting, we study the following three queueing/scheduling schemes employed at the server:

\begin{figure}[tb]
	\centering{
		\includegraphics[width=0.555\linewidth]{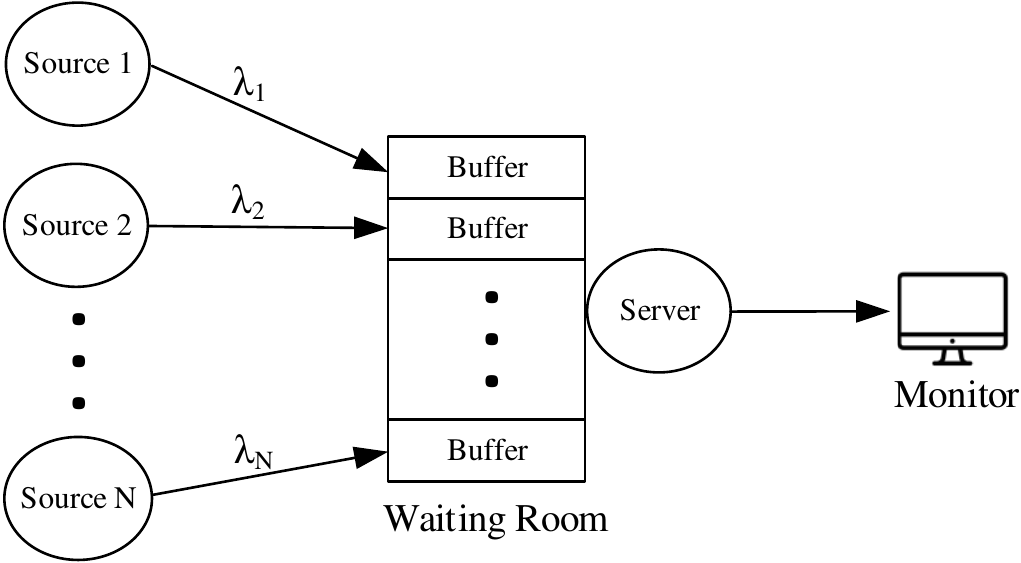}
	}
	\caption{Per-source buffering system where a remote monitor is updated by $N$ information sources through a single server}
	\label{figure:system}
\end{figure}

\begin{itemize}
	%	\item In pure First Source First Serve (pFSFS) and Last Come First Serve (pLCFS) policies, the server selects a packet to serve considering only the time stamps of the packets in the waiting room irrespective of their sources. In pFSFS (pLCFS) policy, the source with the stalest (freshest) packet in the waiting room is selected for service.
	
	\item In the First Source First Serve (FSFS) policy for SBPSQ, also studied in \cite{PacketMan} for the case of 2 sources and a focus on average AoI only, is similar to a FCFS (First Come First Serve) system except that when a new packet arrival belonging to source-$n$ replaces a staler packet in the queue, the service order of that particular source stays the same. If the source-$n$ packet finds its source buffer empty, then its service order will follow the other sources in the waiting room as in FCFS.
	
	%	\item In the Random Service (RS) policy, the server selects a source to serve uniformly at random without paying attention to time stamps.
	
	%Selection probability is same for all sources, i.e., if there exist $M$ packets in the waiting room, each source is selected with probability $1/M$ for $m = 1,2,\dots,M$.
	\item We propose the Earliest Served First Serve (ESFS) policy for SBPSQ for which the server selects a source (with an existing packet in the queue) that has not received service for the longest duration since the previous selection instant of that particular source. In the ESFS policy, the server locally holds an ordered list of sources based on their last selection instants for service. While choosing a source to serve, this ordered list is the only input for the ESFS policy in contrast with the age-based approaches that take into account of the time stamps of information packets in the queue or the instantaneous AoI values at the destination.
	\item For benchmarking purposes, we also consider a server with a one-packet buffer shared by all sources, that is studied as LCFS (Last Come First Serve) with \emph{preemption only in waiting} (LCFS-W) policy in \cite{ref_22} but with emphasis only on the average AoI. In this setting, a packet waiting in the buffer is replaced with a new packet arrival from any source. Upon a service completion, the packet held in the buffer (if it exists) starts to receive service. In our paper, we refer to this scheduling-free queueing policy as the Single Buffer with Replacement (SBR) policy.
	
\end{itemize}

The main contributions of this paper are the following:

\begin{itemize}
	\item We introduce a unifying framework based on Markov fluid queues (MFQ) to numerically obtain the exact per-source distributions of the AoI processes in matrix exponential form for FSFS, ESFS, and SBR policies for general $N$. However, the sizes of the matrices involved in the formulation increase exponentially with $N$ with the first two policies.
	
	%	\item We provide a MATLAB implementation that obtains the exact per-source AoI distributions for each of the three policies using the proposed analytical model. In this implementation, the number of sources, per-source arrival and service rates are input by the user.
	
	\item We study and compare the performance of the three policies under several system load scenarios where the sources may have different traffic intensities but a common service time. Through numerical examples, we show that the proposed age-agnostic ESFS policy, which is quite easy to implement, outperforms the FSFS and SBR policies in terms of the average AoI and the age violation probability averaged across all sources, i.e., symmetric source AoI requirements.

\end{itemize}

The remainder of this paper is organized as follows. In Section~\ref{section:RelatedWord}, related work is given. In Section~\ref{section:MarkovFluidQueues}, MFQs are briefly presented. In Section~\ref{section:SystemModel}, we formally describe the proposed analytical method for obtaining the exact per-source distribution of the AoI process for all three studied policies. In Section~\ref{section:NumericalExamples}, the proposed method is verified with simulations and a comparative analysis of the policies is provided under several scenarios. Finally, we conclude in Section~\ref{section:Conclusions}.

\section{Related Work}
\label{section:RelatedWord}
There have been quite a few studies on queueing-theoretic AoI analysis for multi-source setups when the updates have random service times. The first study on multiple sources sharing a single queue appeared in \cite{ref_66} where the authors derived the average AoI for an M/M/1 FCFS queue. This work is extended in \cite{ref_22} in which the authors studied an M/M/1 queue FCFS service as well as LCFS queues under preemptive and non-preemptive with replacement policies using the SHS (Stochastic Hybrid Systems) technique. A non-preemptive M/M/1/$m$ queue is revisited in \cite{ref_71} where the authors obtained the average AoI expressions. The authors of \cite{ref_69}  independently derived the average AoI for the M/M/1 FCFS model studied in \cite{ref_66} and also provided approximate expressions for a multi-source M/G/1 queue.
In\cite{ref_67}, the peak AoI was studied for multi-source M/G/1 and M/G/1/1 systems with heterogeneous service times. The authors in \cite{ref_68} derived closed form expressions for the average AoI and peak AoI in a multi-source M/G/1/1 queue by extending the single source age analysis in \cite{ref_70}. The authors of \cite{moltafet_tcom} considered three source-aware packet management policies in a two-source system for which they obtained the per-source average AoI for each policy using SHS. The reference \cite{sennur_green} investigated a multi-source status updating system for which the multiple threshold-based scheduling policies along with the closed form expressions for the AoI have been derived. In another line of studies \cite{ref_72},\cite{ref_73},\cite{ref_74}, the AoI analysis of multiple sources with different priorities has been considered under several packet management policies. For distributional properties, the authors in \cite{dhillon_closedFormMGF} studied non-preemptive and preemptive policies for which the moment generating function (MGF) of AoI is obtained using SHS framework. In \cite{moltafet_intensive}, the authors considered the preemptive and blocking policies in a bufferless two-source system deriving the per-source MGF of AoI. The authors of \cite{dogan_akar_tcom21} studied the distributions of both the AoI and peak AoI in a preemptive bufferless multi-source M/PH/1/1 queue allowing arbitrary and probabilistic preemptions among sources.

The most relevant existing studies to this paper are the ones that study the analytical modeling of SBPSQ systems. The benefits of SBPSQ are shown in \cite{pappas_SBPSQ} in terms of lesser transmissions and reduced per-source AoI. In \cite{PacketMan}, the authors obtained the average AoI expressions using SHS techniques for a two-source M/M/1/2 queueing system in which a packet in the queue is replaced only by a newly arriving packet of same source. In \cite{moltafet_SBPSQ}, the authors derived the per-source MGF of AoI in a two-source system for the non-preemptive and self-preemptive policies, the latter being a per-source queueing technique.

\section{Markov Fluid Queues}
\label{section:MarkovFluidQueues}
Markov fluid queues (MFQ) are described by a joint Markovian process $\boldsymbol{X}(t) = (X(t), Z(t))$ with $t \geq 0$ where $X(t)$ represents the fluid level of the process and $Z(t)$ is the modulating Continuous Time Markov Chain (CTMC) with state-space $ \boldsymbol{S} = \{1,2,\dots,K\}$ that determines the rate of fluid change (drift) of the process $X(t)$ at time $t$. The infinitesimal generator of $Z(t)$ is defined as $\boldsymbol{Q} \ (\boldsymbol{\tilde{Q}})$ for $X(t)>0$ ($X(t)=0$) and the drift matrix $\boldsymbol{R}$ is a diagonal matrix with size $K$ which is given as $\boldsymbol{R} = \mathbf{diag}\{r_1,r_2,\dots,r_K\}$ where $r_i$ is the drift value corresponding to the state $i \in \boldsymbol{S}$. When $X(t) = 0$ and $Z(t)=i$ with $r_i < 0$, the fluid level can not be depleted any further, i.e., $X(t)$ sticks to the boundary at zero. The two infinitesimal generators and the drift matrix completely characterize the MFQ, i.e., $\boldsymbol{X}(t) \sim\ MFQ(\boldsymbol{Q}, \boldsymbol{\tilde{Q}},\boldsymbol{R})$, where the size of these matrices, $K$, is the order of the MFQ. In most existing studies, the condition $\boldsymbol{Q} = \boldsymbol{\tilde{Q}}$ is satisfied for which the stationary solutions are obtained in \cite{AnickFluidMF} using the eigendecomposition of a certain matrix. The MFQ models with $\boldsymbol{Q} \neq \boldsymbol{\tilde{Q}}$ turn out to be a special case of multi-regime MFQs whose steady-state solutions can be found using numerically stable and efficient numerical methods as studied in \cite{kankaya_MRMFQ_full}. 

%In this study, the proposed MFQ model satisfies the condition $\boldsymbol{Q} \neq \boldsymbol{\tilde{Q}}$ which turns out to be a special case of multi-regime MFQs whose steady-state solutions can be found using numerically stable and efficient numerical methods as studied in \cite{kankaya_MRMFQ_full}. Moreover, there is a single state in with negative drift in the proposed model for which the solution of interest can be obtained without any numerical method.  
In this paper, we assume that $r_i \neq 0$ for $0 \leq i \leq K$ and there is a single state with unit negative drift which suffices for the AoI models developed in this paper. We consider the case when there are $L = K -1$ states with unit positive drift and $r_i = 1$ ($r_i =-1 $) for $i < K$ ($i = K$) where we particularly defined the state $K$ as the single state with negative drift without loss of generality. Hence, the characterizing matrices of $\boldsymbol{X}(t)$ are written as follows:
\begin{align}
\boldsymbol{Q}=
\left[
\begin{array}{c:c}
\boldsymbol{W} & \boldsymbol{h} \\
\hdashline
\boldsymbol{0} & 0
\end{array}
\right] , 
\boldsymbol{\tilde{Q}}=
\left[
\begin{array}{c:c}
%\boldsymbol{W} & \boldsymbol{h} \\
\boldsymbol{0} & \boldsymbol{0} \\
\hdashline
\boldsymbol{\alpha} & -\boldsymbol{\alpha} \boldsymbol{1}
\end{array}
\right] ,
\boldsymbol{R}=
\left[
\begin{array}{c:c}
\boldsymbol{I} & \boldsymbol{0} \\
\hdashline
\boldsymbol{0} & -1
\end{array}
\right] ,
\label{equation:characterizingMatrices}
\end{align}
where the sizes of the north-west, north-east, and south-west partitions are $L \times L$, $L \times 1$ and $1 \times L$, respectively, and the notations $\boldsymbol{I}$, $\boldsymbol{1}$ and $\boldsymbol{0}$ are used to denote an identity matrix, column matrix of ones, and a matrix of zeros of appropriate sizes, respectively. 
We are interested in finding the steady-state joint probability density function (pdf) vector $\boldsymbol{f_L}(x)$ defined as:
\begin{align}
% F_i^{(k)}(x,t) = Pr\{X(t)\leq x,\: Z(t)=i\}, & 0\leq i\leq N-1\\
% F^{(k)}(x,t) = \left[F_0^{(k)}(x,t) \cdots F_{N-1}^{(k)}(x,t)\right], & T^{(k-1)}< x< T^{(k)}\\ 
f_i(x) & =  \lim\limits_{t\to \infty } \dv{x}  \Pr\{X(t)\leq x,\: Z(t)=i\},   \\
\boldsymbol{f_L}(x) & =  \begin{bmatrix} f_1(x) & f_2(x) & \cdots & f_{K-1}(x) \end{bmatrix}, \label{density}
\end{align}
that is the joint pdf vector containing the states with positive drift. The following theorem provides an expression for the steady-state joint pdf vector $\boldsymbol{f_L}(x)$.
\begin{theorem}
	Consider the process $\boldsymbol{X}(t) \sim\ MFQ(\boldsymbol{Q}, \boldsymbol{\tilde{Q}},\boldsymbol{R})$ with the characterizing matrices as defined in \eqref{equation:characterizingMatrices}. Then, the steady-state joint pdf vector $\boldsymbol{f_L}(x)$ is given in matrix exponential form up to a constant as follows:
	\begin{equation}
	\boldsymbol{f_L}(x) = \eta \boldsymbol{\alpha} \mathrm{e}^{\boldsymbol{W} x} ,\\
	\label{eq:closedForm}
	\end{equation}
	where $\eta$ is a scalar constant. 
\end{theorem}
\begin{proof}
	Let us express the steady-state joint pdf vector of $\boldsymbol{X}(t)$ as $ \boldsymbol{f}(x) = \begin{bNiceArray}{c:c} \boldsymbol{f_L}(x) & f_K(x) \end{bNiceArray}$. Based on \cite{kankaya_MRMFQ_full}, the following differential equation holds for the joint pdf vector $ \boldsymbol{f}(x)$:
	\begin{align}
	\dv{x} \boldsymbol{f}(x) & =  \nonumber
	\begin{bNiceArray}{c:c}
	\boldsymbol{f_L}(x) & f_K(x) \\
	\end{bNiceArray}
	\boldsymbol{Q} \boldsymbol{R}^{-1}, \\ 
	& = \begin{bNiceArray}{c:c}
	\boldsymbol{f_L}(x) & f_K(x) \\
	\end{bNiceArray}
	\left[
	\begin{array}{c:c}
	\boldsymbol{W} & \boldsymbol{-h} \\
	\hdashline
	0 & 0
	\end{array}
	\right],
	\label{equation:diffEquation}
	\end{align}
	along with the following boundary condition also given in \cite{kankaya_MRMFQ_full}:	
	\begin{align}
	\begin{bNiceArray}{c:c}
	\boldsymbol{f_L}(0) & f_K(0^+) \\
	\end{bNiceArray}
	& = 
	\begin{bNiceArray}{c:c}
	\boldsymbol{0} & \eta
	\end{bNiceArray}
	\boldsymbol{\tilde{Q}} \boldsymbol{R}^{-1} , \\
	& =
	\begin{bNiceArray}{c:c}
	\eta \boldsymbol{\alpha} & \eta \boldsymbol{\alpha} \boldsymbol{1}
	\end{bNiceArray}
	,
	\label{equation:boundaryConditions}
	\end{align}
	where $\eta =  \lim_{t\to \infty } \Pr\{X(t)=0,\; Z(t)=K\}$ is the steady-state probability mass accumulation at 0 when $Z(t) = K$. The solution of interest to \eqref{equation:diffEquation} can be written as $\boldsymbol{f_L}(x) = \boldsymbol{f_L}(0) \mathrm{e}^{\boldsymbol{W} x}$ where  $\boldsymbol{f_L}(0)=\eta \boldsymbol{\alpha}$ from \eqref{equation:boundaryConditions}, which completes the proof.
\end{proof}
\begin{remark}
	In \cite{dogan_akar_tcom21}, the scalar constant $\eta$ was also explicitly obtained for similar MFQs with a more elaborate algorithm. 
	However, we have recently observed that obtaining the quantity $\boldsymbol{f_L}(x)$ up to a scalar constant is sufficient for finding the AoI distributions of interest.
\end{remark}

\section{Analytical Models}
\label{section:SystemModel}
We consider the information update system shown in Fig.~\ref{figure:system} consisting $N$ sources with independent arrival and service rates. Packet arrivals are assumed to be Poisson process with traffic intensity vector $(\lambda_1, \lambda_2, \dots, \lambda_N)$ and service times are exponentially distributed with rate vector $(\mu_1, \mu_2, \dots, \mu_N)$ where the per-source load is defined as $\rho_n = \lambda_n / \mu_n$ and the overall system load is given by $\rho = \sum_{n=1}^{N} \rho_n$. The packet management policy is as follows: A newly arriving packet immediately receives service if the server is found idle. Otherwise, the packet gets queued in the 1-packet buffer allocated to that particular source. If the buffer is not empty, the existing packet is replaced only if the arriving packet belongs to same source. Upon a service completion, if there exists only one packet in the waiting room, this packet immediately starts to receive service. On the other hand, a specific policy is applied to select a source to be served if there exist multiple packets in the waiting room. In this setting, which we refer as Single Buffer Per-Source Queueing (SBPSQ), we first study two policies, namely the First Source First Serve (FSFS) and the Earliest Served First Serve (ESFS), for which we construct a unifying MFQ model to obtain the exact AoI distributions for each source. Subsequently, this framework is employed to study the SBR policy.

\subsection{First Source First Serve (FSFS) Policy}
In the FSFS policy, the source of the first packet arrived to the system is the first source to be served. In other words, the service order of sources with an existing packet in the queue is solely determined by their first packet arrival times and thus the service order does not change under replacement events. In our modeling approach, we focus on a source, say source-1, for which we obtain the exact distribution of AoI where the distribution for any source can be obtained similarly by renumbering the sources.

As the first step, we will obtain the probability distribution of the possible system states that an arriving source-1 packet finds upon its arrival to the system (which will be subsequently used while constructing the proposed MFQ model in the second step). For this purpose, we construct a finite state-space Continuous-time Markov chain (CTMC), denoted by $Y(t)$. We enumerate each state for $Y(t)$ as a tuple $q = (i,(P_m)) \in \mathcal{Q}_Y$ where  $ i \in \mathcal{I}_Y = \{0,1,\dots,N\}$ enumerates the source tag of packet that is currently being served except the case when $i=0$ which is used for enumerating the idle server. Let $P_m = s_1,s_2, \dots ,s_m$, $1 \leq m \leq N$, enumerates an $m$-permutation of set $\mathcal{N} = \{1,2,\dots,N\}$ such that any $P_m \in \Gamma_Y$ can be generated by choosing $m$ distinct source tags $s_j$, $1 \leq j \leq m$, from set $\mathcal{N}$ and ordering them. When the server is busy and the queue contains $m$ packets, we define the queue status $(P_m)$ as follows: 
\begin{equation}
(P_m)=
\begin{cases}
%	0, & k=0 \\
(0), & m = 0, \\
(s_1,s_2, \dots ,s_m), & 1 \leq m \leq N, \\
% \infty, & k=K
\end{cases}
\label{eq:MTP_boundary}
\end{equation}  
where the term $(s_1,s_2, \dots ,s_m)$ enumerates the ordering of $m \geq 1$ sources in the queue with respect to their first packet arrival times in ascending order. When there are $m \geq 1$ packets in the queue and a packet belonging to source-$s_j, 1 \leq j \leq m$, arrives to the system, the replacement event occurs but the queue status $(P_m)$ does not get updated. According to the FSFS policy, the packet of leftmost source will receive service first, i.e., $s_1$ denotes the source tag of packet which will receive service first among those in the queue. Similarly, $s_2$ is the source tag of packet which will receive service after the service completion of source-$s_1$ and so on. Since, the $s_j$ terms $\forall j \in \{1,2,\dots,m\}$ in $P_m$ are all distinct, we also denote the set of sources with an existing packet in the queue as $\{P_m\}$ without any ambiguity. Lastly, when the server is idle, we enumerate the system state as $q = (0,(0))$ since there cannot be any waiting packet in the queue when the server is idle. 

Suppose that the system state at time $t$ is $Y(t) = q$ at which moment a service completion event occurs when there are $m > 0$ packets in the queue. According to the FSFS policy, the server selects the packet of source-$s_1$ for service after which the system transitions into the state $q' = (s_1,(P'_m))$ where the updated queue status $(P'_m)$ with $m-1$ packets in the queue is given as:
\begin{equation}
(P'_m)=
\begin{cases}
%	0, & k=0 \\
(0), & m = 1, \\
(s_2,s_3, \dots ,s_m), & 1 < m \leq N, \\
% \infty, & k=K
\end{cases}
\label{eq:MTP_boundary}
\end{equation}  
that is the source-$s_1$ is removed from the ordered list of sources with an existing packet in the queue.

Let $\nu_{q,q'} > 0$, $q,q' \in \mathcal{Q}_Y$, denotes the rate of transition from state $q = (i,(P_m))$ to state $q' \neq q$ where we list all such transitions for the FSFS policy in Table~\ref{Table:TransitionRates_FSFS} for which the rows 1-3 (4-6) correspond to the rates for arrival (departure) events. For any other state pair $q$, $q' \in \mathcal{Q}_Y$, the transition rate $\nu_{q,q'}$ is zero.

\begin{table}[h]
	\centering
	\caption{Transition rates $\nu_{q,q'}$ of $Y(t)$ for the FSFS policy}
%	\scalebox{1}{
		\begin{tabular}{|c|c|c|c|}
			\hline
			$q$ 		   & $q'$                 						& $\nu_{q,q'}$  & Condition 		\\ \hline
			$(0,(0))$        & $(i,(0))$                          			& $\lambda_i$	  & $i \in \mathcal{N}$			 	\\ \hline
			$(i,(0))$        & $(i,(j))$                          			& $\lambda_j$	  & $i, j \in \mathcal{N}$			 	\\ \hline
			$(i,(P_m))$        & $(i,(P_m,j))$                          			& $\lambda_j$	  & $i,j \in \mathcal{N}, $ 
			\\ & & & $j\notin \{P_m\}, P_m \in \Gamma_Y$			 	\\ \hline % j notin Qm kurtariyor.
			$(i,(0))$        & $(0,(0))$                          			& $\mu_i$	  & $i \in \mathcal{N}$			 	\\ \hline
			$((i,(P_m))$        & $(s_1,(P'_m))$                          			& $\mu_i$	  & $i \in  \mathcal{N}, P_m \in \Gamma_Y$			 	\\ \hline
			%		$(i,(P_m))$              & $(s_1,(s_2,\dots,s_m))$                          		& $\mu_i$				  & $i \in \mathcal{N}, P_m \in \Gamma_Y, m > 1$			 	\\ \hline
		\end{tabular}
%	}
	\label{Table:TransitionRates_FSFS}
\end{table}

Let us denote the probability that the system is in state $q$ as time goes to infinity, i.e., $ \pi_q = \lim_{t \to \infty} P(Y(t) = q)$. Following the ergodicity of Markov chain $Y(t)$, the steady-state distribution converges to a unique vector consisting of elements $\pi_q, q \in \mathcal{Q}_Y$, which satisfies the following set of linear equations:
\begin{align}
& \pi_q \sum_{q' \in \mathcal{Q}_Y} \nu_{q,q'} = \sum\limits_{q' \in \mathcal{Q}_Y} \pi_{q'} \nu_{q',q}, \forall q \in \mathcal{Q}_Y, \label{Yt_equation_1} \\
& \sum_{q \in \mathcal{Q}_Y} \pi_q = 1. \label{Yt_equation_2}
\end{align}
Since the packet arrivals are Poisson, the probability that an arriving packet finds the system in state $q \in \mathcal{Q}_Y$ is $ \pi_q$ as a direct consequence of the \emph{PASTA} (Poisson Arrivals See Time Averages) property \cite{gross_harris_book}.

%Note that we can further simplify the obtained steady-state distribution $\boldsymbol{\pi}$ since some states are identical in view of an arriving source-1 packet under replacement policy. In particular, when the server is busy, an arriving source-1 packet may find its buffer either occupied or empty. In either case the new packet will be kept in the buffer under replacement policy. Using this approach, $\pi^{(i,j)}$
%
%we write the simplified steady-state distribution as $1 \times  5$ row vector, denoted by $\boldsymbol{\pi'}$, as follows: $\pi'_1 = \pi_3 + \pi_5,\: \pi'_2 = \pi_2 + \pi_4, \: \pi'_3 = \pi_7 + \pi_9, \: \pi'_4 = \pi_6 + \pi_8, \: \pi'_5 = \pi_1$ where the order of states are selected in a way that the vector $\boldsymbol{\pi'}$ will directly be embedded into the matrix $\boldsymbol{\tilde{Q}}$ as detailed in the MFQ modeling.  

In the second step, we construct the proposed MFQ process $\boldsymbol{X}(t) = (X(t), Z(t))$ which describes a fluid level trajectory with infinitely many independent cycles as shown in Fig.~\ref{figure:samplePath} where each cycle begins with an arriving source-1 packet to the system and ends with either the reception of the next source-1 packet at the destination (cycle 3 and 5 in Fig.~\ref{figure:samplePath}) or the possible packet replacement by another source-1 arrival (cycle 1, 2 and 4 in Fig.~\ref{figure:samplePath}). First, we construct the state-space $\boldsymbol{S}$ of sub-process $Z(t)$ by dividing a cycle into four phases and defining the set of states for each phase. For state enumerations, we define three additional tags for packets belonging to source-1 to differentiate between them in different states and phases: 
\begin{itemize}
	\item The packet $1_c$, i.e., current source-1 packet, refers to the source-1 packet that initiates each cycle with its arrival to the system.
	\item When a packet $1_c$ arrives to the system, the server can be busy already serving another source-1 packet which is enumerated as $1_p$, i.e., previous source-1 packet.
	\item The packet $1_n$ (next source-1 packet) enumerates the received source-1 packet subsequent to the packet $1_c$ at destination.
\end{itemize}
\begin{figure*}[tbh]
	\centering{
		\includegraphics[width=1\linewidth]{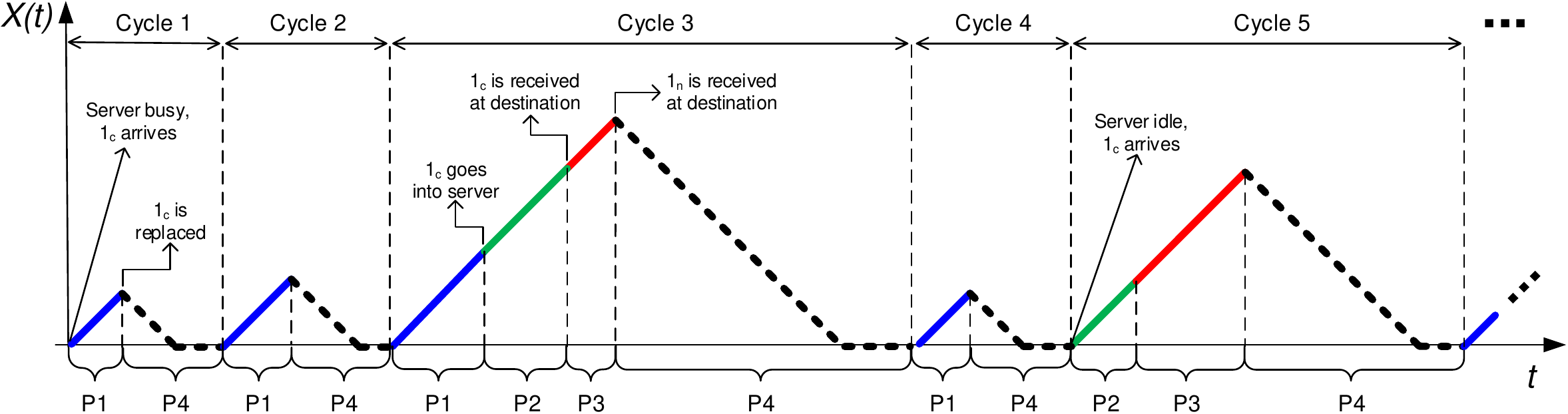}
	}
	\caption{Sample path of the fluid level process $X(t)$ with infinitely many independent cycles.}
	\label{figure:samplePath}
\end{figure*}
Each cycle consists of four phases, namely Phase 1-4 as shown in Fig.~\ref{figure:samplePath}. If the server is busy when the packet $1_c$ arrives, the cycle starts from Phase 1 (solid blue curve) during which the fluid level $X(t)$ increases at a unit rate and terminates with either the beginning of its service time at which moment the system transitions into Phase 2 (which occurs during cycle 3 in Fig.~\ref{figure:samplePath}) or the replacement of it by another source-1 arrival (which occurs during cycle 1, 2 and 4 in Fig.~\ref{figure:samplePath}). In the latter case, the queue wait time of the packet $1_c$ needs a reset which is accomplished by transitioning directly into a final phase, defined as Phase 4 (dashed black curve), that is used for reseting the fluid level by bringing it down to zero with a unit rate. For Phase 1, we enumerate each state as $q = (i,(P_m)) \in \mathcal{Q}_1$ where  $ i \in \mathcal{I}_1 = \{1_p,2,\dots,N\}$ denotes the source tag of the packet under service. For any $i \in \mathcal{I}_1$ value, the element $(P_m)$, $1 \leq m \leq N$, enumerates the ordering of packets in the queue similar to previously given definition for $Y(t)$ with the exception that the packet $1_c$ always exists in the queue during Phase 1. Thus, any $P_m \in \Gamma_1$ can be generated by ordering $1_c$ and another $(m-1)$ distinct source tags selected from the set $\{2,3,\dots,N\}$. With all these definitions, we enumerate the queue status $(P_m)$ containing $m$ packets for Phase 1 as follows:
\begin{equation}
(P_m)= (s_1,s_2, \dots ,s_m), \; 1_c \in \{P_m\}, \; 1 \leq m \leq N, \\
\label{eq:MTP_boundary}
\end{equation}
which is valid for any $i \in \mathcal{I}_1$ server status.

In addition to possible transition from Phase 1, if an arriving packet $1_c$ finds the system idle, a direct transition to Phase 2 (solid green curve) occurs (which is shown as cycle 5 in Fig.~\ref{figure:samplePath}). During Phase 2, the fluid level continues to increase at a unit rate until the reception instant of packet $1_c$ at destination at which moment the system transitions into Phase 3. Note that once the packet $1_c$ goes into server, it can no longer be replaced by another packet arrival. Thus, the only possible transition out of Phase 2 is into Phase 3. We enumerate each state for Phase 2 as $q = (1_c,(P_m)) \in \mathcal{Q}_2$ where the queue status $(P_m)$, $1 \leq m \leq N$, is similar to previously given definitions with the exception that the packet $1_n$ may exist or not in the queue during Phase 2. In the latter case, any $P_m \in \Gamma_2$, $1 \leq m \leq N-1$, can be generated by ordering $m$ distinct source tags selected from the set $\{2,3,\dots,N\}$. In the former case, we impose a restriction given as $s_m = 1_n$, i.e., the last packet to be served in the queue is always $1_n$. The reason is that any packet behind $1_n$ in the queue is irrelevant because, as shown in Fig.~\ref{figure:samplePath}, the system always transitions into the final phase, i.e., Phase 4, at the reception instant of packet $1_n$ at destination regardless of the queue status. Therefore, for this case, any $P_m \in \Gamma_2$, $1 \leq m \leq N$, can be generated by selecting $1_n$ and another $(m-1)$ distinct source tags from set $\{2,3,\dots,N\}$, and ordering them while satisfying the condition $s_m = 1_n$, i.e., the last source to be served is the source-1 when there are $m \geq 1$ packets in the queue. Finally, when there is no packet in the queue, we define the queue status as $(P_m) = (0)$. With all these definitions, we enumerate the queue status $(P_m)$ containing $m$ packets for Phase 2 as follows:
\begin{equation}
(P_m)=
\begin{cases}
%	0, & k=0 \\
(0), & m = 0, \\
(s_1,s_2, \dots ,s_m), & 1 \leq m \leq N-1, 1_n \notin \{P_m\}, \\
(s_1,s_2, \dots ,s_m), & 1 \leq m \leq N, s_m = 1_n. \\
% \infty, & k=K
\end{cases}
\label{eq:Phase2_Q}
\end{equation}
Once Phase 2 is over, Phase 3 (solid red curve) starts and continues until the reception of the packet $1_n$, at destination. Each state for Phase 3  is enumerated as $q = (i,(P_m)) \in \mathcal{Q}_3$ where  $ i \in \mathcal{I}_3 = \{0, 1_n,2,\dots,N\}$ denotes the source tag of the packet under service except the case when $i=0$ which is used for enumerating the idle server. Similar to the arguments for Phase 2, any packet behind $1_n$ in the system is irrelevant. Therefore, when the packet under service is $1_n$, the system state is enumerated as $q = (1_n,(0))$. If the server is busy but the packet under service is not $1_n$, i.e., $i \neq 1_n$, the buffer status $(P_m)$ can be defined as given in (\ref{eq:Phase2_Q}) similar to Phase 2. In particular, if the buffer contains the packet $1_n$, any $P_m \in \Gamma_3$, $1 \leq m \leq N$, can be generated by ordering $1_n$ and another $(m-1)$ distinct elements selected from set $\{2,3,\dots,N\}$, satisfying the condition $s_m = 1_n$. If the buffer does not contain $1_n$, any $P_m \in \Gamma_3$, $1 \leq m \leq N-1$, can be generated by ordering $m$ distinct elements selected from the set $\{2,3,\dots,N\}$. Finally, when $i=0$, we enumerate the idle server status as $q = (0,(0))$ which may occur only in Phase 3 when the packet $1_c$ was delivered to the destination but the next source-1 packet, i.e., packet $1_n$, has not yet arrived to the system.

%\begin{equation}
%(P_m)=
%\begin{cases}
%%	0, & k=0 \\
%(0), & m = 0, \\
%(s_1,s_2, \dots ,s_m), & 1 \leq m \leq N-1, 1_n \notin \{P_m\}, \\
%(s_1,s_2, \dots ,s_m), & 1 \leq m \leq N, s_m = 1_n. \\
%% \infty, & k=K
%\end{cases}
%\label{eq:MTP_boundary}
%\end{equation}

Once Phase 3 is over, the system transitions into the final stage, i.e., Phase 4, where the fluid level is brought down to zero with a drift of minus one after which the fluid level stays at zero for exponentially distributed time with unit rate. Thus, Phase 4 consists of a single state which we enumerate as $q = (-1,(-1)) \in \mathcal{Q}_4$. After the fluid level is brought down to zero in Phase 4, the only possible transition out of Phase 4 is to Phase 1 or 2 both of which initiates a new cycle that is independent from all previous cycles.  With all these definitions, the state-space $\boldsymbol{S}$ of $Z(t)$ can now be defined as $\boldsymbol{S} =  \bigcup_{p=1}^{4} \mathcal{Q}_p$ consisting of all states defined for Phase 1-4.

State transition diagram of the subprocess $Z(t)$ can be represented as a directed graph as shown in Fig.~\ref{Figure:EvolutionPhase} where each edge represents a set of transitions between or within phases. We define the corresponding transition rates such that if the system remains in the same phase after a transition, we will refer such transition as intra-phase transition for which the rate is denoted as $\alpha_{q,q'} , q,q' \in \mathcal{Q}_p, p = 1,2,3,4$, whereas if it enters to another phase, it will be referred as inter-phase transition in which case the rate is denoted as $\beta_{q,q'}, q \in \mathcal{Q}_p, q' \notin \mathcal{Q}_p, p = 1,2,3,4$. For the FSFS policy, all intra-phase and inter-phase transitions are listed in Table~\ref{Table:intraPhase_FSFS} and Table~\ref{Table:interPhase_FSFS}, respectively, where the set $\mathcal{J}_p$, $p=1,2,3$ is defined as the set of source tags to which any packet in the queue may belong in Phase $p$ that is $\mathcal{J}_1 = \{1_c,2,\dots,N\}$ and $\mathcal{J}_2 = \mathcal{J}_3 = \{1_n,2,\dots,N\}$. Unless explicitly stated in the corresponding row, given transitions are defined for the condition $X(t) > 0$, which constitute the entries of matrix $\boldsymbol{Q}$, whereas the transitions defined for $X(t) = 0$ constitute the entries of the matrix $\boldsymbol{\tilde{Q}}$. For intra-phase transitions, the rows 1-2, 3-4 and 5-9 refer to the transitions for Phase 1, 2 and 3, respectively. Note that there is no intra-phase transition for Phase 4 since its state-space consists of a single state. For inter-phase transitions, the rows 1, 2, 3-4, 5, 6 and 7-9 refer to the transitions from Phase 1 to 2, Phase 1 to 4, Phase 2 to 3, Phase 3 to 4, Phase 4 to 2 and Phase 4 to 1, respectively. Since the transitions from Phase 4 to 1 or 2 initiate a new cycle, their rates are proportional to the steady-state distribution of the system status that a source-1 packet finds upon its arrival to the system. By solving the steady-state distribution of the process $Y(t)$ as described in the first step, the rates of these transitions are expressed as given in the last three rows of Table~\ref{Table:interPhase_FSFS}. Expressing the transition rates in terms of the steady-state probabilities of $Y(t)$ stems from the fact that the fluid level stays at zero in Phase 4 for exponentially distributed time with unit rate, i.e., the sum of transitions out of Phase 4 when $X(t) = 0$ should be exactly one which equals to the sum of steady-state probabilities $\pi_q, q \in \mathcal{Q}_Y$.
\begin{figure}
	\centering
	\begin{tikzpicture}[->, >=stealth', auto, semithick, node distance=3.8cm]
	\tikzstyle{every state}=[fill=white,draw=black,thick,text=black,scale=1]
	\node[state]    (n1)                     {\small $q \in \mathcal{Q}_1$};
	\node[state]    (n2)[right of=n1]   		{\small $q \in \mathcal{Q}_2$};
	\node[state]    (n4)[below of=n1]   		{\small $q \in \mathcal{Q}_4$};
	\node[state]    (n3)[below of=n2]   		{\small $q \in \mathcal{Q}_3$};
	\path
	(n1) edge[loop left]		node{}	(n1)
	edge     				node[text width=1.7cm,align=center]{\small $1_c$ goes into server}     (n2)
	edge    node[text width=1.2cm,align=left]{\small $1_c$ is replaced}     (n4)
	(n2) edge[loop right]		node{}	(n2)
	edge     				node[text width=2cm,align=center]{\small $1_c$ is received at destination}     (n3)
	(n3) edge[loop right]		node{}	(n3)
	edge     				node[text width=2.2cm,align=center]{\small $1_n$ is received at destination}     (n4)
	(n4) edge[bend left]    node[text width=1.2cm,align=center]{\small Server busy, $1_c$ arrives}     (n1)
	edge     node[text width=1.5cm,align=center, below,sloped]{\small Server idle, $1_c$ arrives}     (n2);
	
	%\node[above=0.5cm] (A){Patch G};
	%\draw[red] ($(D)+(-1.5,0)$) ellipse (2cm and 3.5cm)node[yshift=3cm]{Patch H};	
	\end{tikzpicture}
	\caption{State transition diagram of the subprocess $Z(t)$}
	\label{Figure:EvolutionPhase}
	
\end{figure}
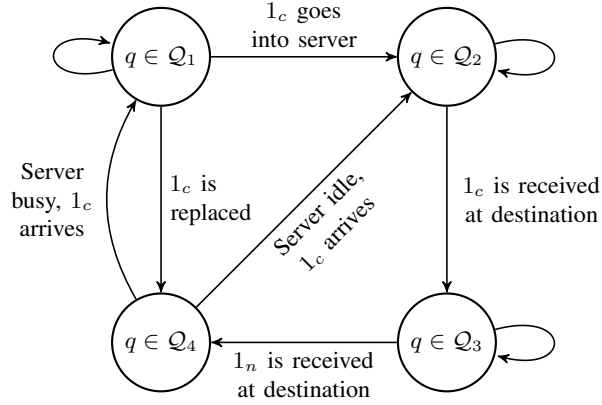  
\begin{table}[!htb]
	\centering
	\caption{Intra-phase transition rates $\alpha_{q,q'}$ for the FSFS policy}
%	\scalebox{0.96}{
		\begin{tabular}{|c|c|c|c|}
			\hline	
			$q$ 		   & $q'$                 						& $\alpha_{q,q'}$  & Condition 		\\ \hline
			
			%% Phase 1
			$(i,(P_m))$        & $(i,(P_m,j))$                          			& $\lambda_j$	  & $i \in \mathcal{I}_1, j \in \mathcal{J}_1$
			\\ & & & $j\notin \{P_m\}, P_m \in \Gamma_1$	\\ \hline		 	
			
			$(i,(P_m))$              & $(s_1,(P'_m))$                          		& $\mu_i$				  & $i \in \mathcal{I}_1,s_1 \neq 1_c$		 	
			\\ & & & $P_m \in \Gamma_1$ \\ \hline
			
			%% Phase 2
			$(1_c,(0))$        & $(1_c,(j))$                          			& $\lambda_j$	  & $j \in \mathcal{J}_2$ \\ \hline
			
			$(1_c,(P_m))$        & $(1_c,(P_m,j))$                          			& $\lambda_j$	  & $j \in \mathcal{J}_2, j\notin \{P_m\},$
			\\ & & & $ s_m \neq 1_n, P_m \in \Gamma_2$	\\ \hline		 	 
			
			%% Phase 3
			$(0,(0))$        & $(i,(0))$                          			& $\lambda_i$	  & $i \in \mathcal{I}_3-\{0\}$			 	\\ \hline
			
			$(i,(0))$        & $(i,(j))$                          			& $\lambda_j$	  & $i \in \mathcal{I}_3-\{0, 1_n\}, j \in \mathcal{J}_3$	\\ \hline		 	
			
			$(i,(P_m))$        & $(i,(P_m,j))$                          			& $\lambda_j$	  & $i \in \mathcal{I}_3-\{0, 1_n\}, j \in \mathcal{J}_3,$
			\\ & & & $s_m \neq 1_n, j\notin \{P_m\}, P_m \in \Gamma_3$ \\ \hline		 	
			
			$(i,(0))$        & $(0,(0))$                          			& $\mu_i$	  & $i \in \mathcal{I}_3-\{0, 1_n\}$		 	\\ \hline
			
			$(i,(P_m))$        & $(s_1,(P'_m))$                          			& $\mu_i$	  & $i \in \mathcal{I}_3-\{0, 1_n\}, P_m \in \Gamma_3$ \\ \hline			 	
			
			%		$(i,(P_m))$        & $(s_1,(0))$                          			& $\mu_i$	  & $i \in \mathcal{I}_3-\{0\},$			 	
			%		\\ & & & $P_m \in \Gamma_3, m = 1$ \\ \hline
			%		
			%		$(i,(P_m))$              & $(s_1,(s_2,\dots,s_m))$                          		& $\mu_i$				  & $i \in \mathcal{I}_3-\{0, 1_n\},$		 	
			%		\\ & & & $P_m \in \Gamma_3, m > 1$ \\ \hline

		\end{tabular}
%	}
	\label{Table:intraPhase_FSFS}
\end{table}
%\cellcolor{green}
\begin{table}[!htb]
	\centering
	\caption{Inter-phase transition rates $\beta_{q,q'}$ for the FSFS policy}
	\begin{tabular}{|c|c|c|c|}
		\hline	
		$q$ 		   & $q'$                 						& $\beta_{q,q'}$  & Condition 		\\ \hline
		
		%% Phase 1 -> 2
		$(i,(P_m))$              & $(1_c,(P'_m))$                          		& $\mu_i$				  & $i \in \mathcal{I}_1, s_1=1_c,$		 			
		\\ & & & $P_m \in \Gamma_1$ \\ \hline
		
		%		$(i,(P_m))$              & $(1_c,(s_2,\dots,s_m))$                          		& $\mu_i$				  & $i \in \mathcal{I}_1,s_1 =  1_c$		 	
		%		\\ & & & $P_m \in \Gamma_1, m > 1$ \\ \hline		
		%		$(i,(P_m))$              & $(1_c,(0))$                          		& $\mu_i$				  & $i \in \mathcal{I}_1$		 	
		%		\\ & & & $P_m \in \Gamma_1, m = 1$ \\ \hline
		
		%% Phase 1 -> 4
		$(i,(P_m))$              & $(-1,(-1))$                          		& $\lambda_1$				  & $i \in \mathcal{I}_1, P_m \in \Gamma_1$ \\ \hline
		
		%% Phase 2 -> 3
		$(1_c,(0))$              & $(0,(0))$                          		& $\mu_1$				  & \\ \hline
		
		$(1_c,(P_m))$        & $(s_1,(P'_m))$                         			& $\mu_1$	  & $P_m \in \Gamma_2$ \\ \hline
		
		%		$(1_c,(P_m))$        & $(s_1,(0))$                          			& $\mu_1$	  & $P_m \in \Gamma_2, m = 1$ \\ \hline		
		%		$(1_c,(P_m))$        & $(s_1,(s_2,\dots,s_m))$                         			& $\mu_1$	  & $P_m \in \Gamma_2, m > 1$ \\ \hline
		
		%% Phase 3 -> 4
		$(1_n,(0))$              & $(-1,(-1))$                          		& $\mu_1$				  & \\ \hline
		
		%% Phase 4 -> 2
		$(-1,(-1))$              & $(1_c,(0))$                          		& $\pi_{(0,(0))}$				  & $X(t) = 0$\\ \hline
		
		%% Phase 4 -> 1
		$(-1,(-1))$              & $(i,(P_m))$                          		& $\pi_{(i,(P_m))}$				  & $X(t) = 0,$
		\\ & & & $i \in \mathcal{I}_1, s_m \neq 1_c,$ 
		\\ & & & $1_c \in \{P_m\}, P_m \in \Gamma_1$ \\ \hline
		
		$(-1,(-1))$              & $(i,(1_c))$                          		& $\pi_{(i,(0))} + $				  & $X(t) = 0, i \in \mathcal{I}_1$
		\\ & & $\pi_{(i,(1))}$ & \\ \hline
		
		$(-1,(-1))$              & $(i,(P_m,1_c))$                          		& $\pi_{(i,(P_m))} + $				  & $X(t) = 0, i \in \mathcal{I}_1,$
		\\ & & $\pi_{(i,(P_m,1))}$ & $1_c \notin \{P_m\}, P_m \in \Gamma_1$ \\ \hline	
	\end{tabular}
	\label{Table:interPhase_FSFS}
\end{table}

Next, we define the drift value $r_q$ for each state $\forall q \in \boldsymbol{S}$ which constitutes the diagonal entries of the drift matrix $\boldsymbol{R}$. Since the fluid level increases at a unit rate in each state for Phase 1-3, we have $r_q = 1, \forall q \in \bigcup_{p=1}^{3} \mathcal{Q}_p$ whereas the fluid level is brought down to zero with a drift of minus one in Phase 4. Thus, we have $r_q = -1, \forall q \in \mathcal{Q}_4$ which completes the construction of the proposed MFQ model $\boldsymbol{X}(t) \sim\ MFQ(\boldsymbol{Q}, \boldsymbol{\tilde{Q}},\boldsymbol{R})$. From state definitions, the MFQ $\boldsymbol{X}(t)$ has a single state with negative drift and its characterizing matrices can be written as in \eqref{equation:characterizingMatrices} by ensuring that the state with negative drift, i.e., $(-1,(-1)) \in \mathcal{Q}_4$, is numbered as the last state that is the state $K$ in the formulation given in Section~\ref{section:MarkovFluidQueues}. 

%For the states with positive drift, the ordering can be arbitrary in the characterizing matrices. 
%The joint steady-state distribution for the states with positive drift only, i.e., $f_q(x), \forall q \in \bigcup_{p=1}^{3} \mathcal{Q}_p$, is then obtained using \eqref{eq:closedForm} up to a scalar constant.

By sample path arguments, we observe that one sample cycle of the AoI process coincides with the part of sample cycle of $X(t)$ associated with Phase 3 only as indicated by the red parts of the curve in Fig.~\ref{figure:samplePath}. Therefore, the probability density function (pdf) of the AoI for source-1, denoted by $f_{\Delta_1}(x)$, can be written as follows:
\begin{align}
f_{\Delta_1}(x) = \frac{\sum\limits_{q \in \mathcal{Q}_3} f_q(x)}{\int\limits_{0}^{\infty} \sum\limits_{q \in \mathcal{Q}_3} f_q(x')\,dx'}, x \geq 0.
\end{align}
For censoring out all states with positive drift other than the ones in $\mathcal{Q}_3$, we define a column vector $\boldsymbol{\beta}$ of size $L$ containing only zeros except for the states $q \in \mathcal{Q}_3$ for which it is set to one. Using \eqref{eq:closedForm} along with this definition, we can finally obtain:
\begin{align}
f_{\Delta_1}(x) = \epsilon \: \boldsymbol{\alpha} \mathrm{e}^{\boldsymbol{W} x} \boldsymbol{\beta}, \: x \geq 0,
\label{equation:finalPDF}
\end{align}
where $\epsilon^{-1} = - \boldsymbol{\alpha} \boldsymbol{W^{-1}}\boldsymbol{\beta}$. The $k$th non-central moments of $\Delta_1$ can also be easily written as follows:
\begin{align}
%E\Big[(\Delta_1)^k \Big] = \epsilon \: \boldsymbol{\alpha} (\boldsymbol{-W})^{-k-1} \boldsymbol{\beta}.
E\Big[(\Delta_1)^k \Big] = (-1)^{k+1} k! \; \epsilon \: \boldsymbol{\alpha} \boldsymbol{W^{-(k+1)}} \boldsymbol{\beta}.
\label{eq:moments}
\end{align}
Similar steps are then followed for obtaining the pdf of the AoI for source-$n$, denoted by $f_{\Delta_n}(x), n = 2, 3, \dots, N$, by renumbering the sources. Finally, we define the performance metrics of interest, namely the average AoI and the average age violation probability, denoted by $E[\Delta]$ and $\Theta(\gamma)$, respectively, as follows:
\begin{equation}
E[\Delta] = \sum\limits_{n=1}^{N} \frac{E[\Delta_n]}{N}, \ \Theta(\gamma) = \sum\limits_{n=1}^{N} \frac{Q_{\Delta_n}(\gamma)}{N},
\end{equation}
%\Theta(\gamma) = \frac{Q^{(1)}_\Delta(\gamma)+Q^{(2)}_\Delta(\gamma)}{2},
where $\Delta= \frac{1}{N} \sum_{n=1}^{N} \Delta_n $, $E[\Delta_n]$ is the average AoI for source-$n$, and $Q_{\Delta_n}(\gamma)$ is the age violation probability for source-$n$ which is calculated as $Q_{\Delta_n}(\gamma) = Pr\{\Delta_n > \gamma\}$ where $\gamma$ is a given age violation threshold.

The framework that we introduced in this subsection is unifying in the sense that it can be generalized to any SBPSQ policy by only redefining the following terms:
\begin{itemize}
	\item The state-space $\mathcal{Q}_Y$ and the corresponding transition rates $\nu_{q,q'}$ of the process $Y(t)$,
	\item The state-space $\mathcal{Q}_p$ for $p=1,2,3,4$, and the corresponding intra (inter) phase transition rates $\alpha_{q,q'}$ ($\beta_{q,q'}$),
\end{itemize}
since the sample path of the fluid level process shown in Fig.~\ref{figure:samplePath} and the state transition diagram shown in Fig.~\ref{Figure:EvolutionPhase} are valid for any such policy. In fact, from sample path arguments, this generalization also holds for the SBR policy. Therefore, using this unifying framework, we provide the analytical models for both ESFS and SBR policies by only redefining the above-mentioned state-spaces and transition rates.

\subsection{Earliest Served First Serve (ESFS) Policy}
Each state for $Y(t)$ is enumerated as a tuple $q = ((H),\{C_m\}) \in \mathcal{Q}_Y$ for the ESFS policy. Let $H = h_1, h_2, \dots, h_N$ enumerates any $N$-permutation of set $\mathcal{N}$ such that any $H$ can be generated by choosing $N$ distinct source tags from set $\mathcal{N}$, i.e., all source tags, and ordering them. Accordingly, the element $(H) = (h_1, h_2, \dots, h_N) \in \mathcal{H}_Y$ is defined as the service status where the sources are listed in descending order with respect to their last service time instants. In other words, the tag $h_1 (h_N)$ indicates the source that has not received service for the longest (shortest) duration. For any state except the idle server, the tag $h_N$ indicates the source whose packet is currently being served. Therefore, when a packet belonging to source-$i$ goes into server, the tag $h_N$ has to be updated as $i$ and the other terms have to be shifted accordingly. For this purpose, we define an operation $\Upsilon(H, i)$ that updates the service status when a packet belonging to source-$i$ goes into server as follows:
\begin{equation}
\Upsilon(H, i) = H'_i = (h_1,\dots,h_{f-1},h_{f+1}, \dots, h_N, h_f),
\end{equation}
where $h_f = i$, i.e., the tag $h_f$ belongs to source-$i$. Furthermore, we let $C_m = s_1,s_2, \dots ,s_m$, $1 \leq m \leq N$, enumerates an $m$-combination of set $\mathcal{N} = \{1,2,\dots,N\}$ such that any $C_m \in \Gamma_Y$ can be generated by choosing $m$ distinct source tags $s_j$, $1 \leq j \leq m$, from set $\mathcal{N}$. Accordingly, the element $\{C_m\} = \{s_1,s_2, \dots ,s_m\}$ is defined as the set of $m \geq 1$ sources with an existing packet in the queue where the ordering of $s_j$ terms is irrelevant in contrast with the FSFS policy. In the ESFS policy, the server selects the packet belonging to the source that has not received service for the longest duration among those with an existing packet in the queue. Suppose that the system state at time $t$ is $Y(t) = q$ at which moment a service completion event occurs when there are $m > 0$ packets in the queue. In line with the ESFS policy, the server selects the packet of source-$i^*$ for service where the tag $i^*$ is defined as:
\begin{equation}
i^* = h_{f^*}, \:  f^* = \min_{\forall f \in \mathcal{N}} f, h_f \in \{C_m\},
\end{equation}
after which the system transitions into the state $q' = ((H'_{i^*}),\{C'_m\})$ where the updated queue status $\{C_m'\}$ with $m-1$ packets in the queue is given as:
\begin{equation}
\{C_m'\}=
\begin{cases}
%	0, & k=0 \\
\{0\}, & m = 1, \\
\{C_m\} - \{i^*\}, & 1 < m \leq N, \\
% \infty, & k=K
\end{cases}
\label{eq:MTP_boundary}
\end{equation}  
that is the source-${i^*}$ is removed from the list of sources with an existing packet in the queue. Next, we define the system states with an empty buffer as follows:
\begin{itemize}
	\item When the server is busy but the queue is empty, we define the system state as $q = ((H),\{0\})$ where the packet in service belongs to source-$h_n$.
	\item When the server is idle, we define the system state as $q = ((H),\{-1\})$ since the service status has to be always preserved in the ESFS policy even if the server is idle. In this case, the source-$h_N$ is the source whose packet has been served most recently but is not currently in service.
\end{itemize}
This concludes the state definitions for the process $Y(t)$ after which we define the transition rates $\nu_{q,q'}$ of $Y(t)$ in Table~\ref{Table:TransitionRates_ESFS} where the rates correspond to the arrival (departure) events are given in the rows 1-3 (4-5).
\begin{table}[!htb]
	\centering
	\caption{Transition rates $\nu_{q,q'}$ of $Y(t)$ for the ESFS policy}
%	\scalebox{1}{
		\begin{tabular}{|c|c|c|c|}
			\hline
			$q$ 		   & $q'$                 						& $\nu_{q,q'}$  & Condition 		\\ \hline
			$((H),\{-1\})$        & $((H'_i),\{0\})$                          			& $\lambda_i$	  & $H \in \mathcal{H}_Y$			 	\\ \hline
			$((H),\{0\})$        & $((H),\{j\})$                          			& $\lambda_j$	  & $H \in \mathcal{H}_Y, j \in \mathcal{N}$			 	\\ \hline
			$((H),\{C_m\})$        & $((H),\{C_m, j\})$                          			& $\lambda_{j}$	  & $H \in \mathcal{H}_Y, j \in \mathcal{N},$
			\\ & & & $j \notin \mathcal{C}_m, \mathcal{C}_m \in \Gamma_Y$ 
			\\ \hline 
			$((H),\{0\})$        & $((H),\{-1\})$                          			& $\mu_{h_N}$	  &	$H \in \mathcal{H}_Y$		 	\\ \hline
			$((H),\{C_m\})$        & $((H'_{i^*}),\{C'_m\})$                          			& $\mu_{h_N}$	  & $H \in \mathcal{H}_Y, {C}_m \in \Gamma_Y$
			\\ \hline
		\end{tabular}
%	}
	\label{Table:TransitionRates_ESFS}
\end{table}
Next, we define the states $q  = ((H),\{C_m\}) \in \mathcal{Q}_p$ for each phase.
\begin{itemize}
	\item For Phase 1, the service status $(H) \in \mathcal{H}_1$ is defined as an $N$-permutation of set $\mathcal{I}_1 = \{1_p,2,\dots,N\}$. For any $(H) \in \mathcal{H}_1$, the packet $1_c$ always exists in the queue from the definition of Phase 1. Thus, any $C_m \in \Gamma_1$, $1\leq m \leq N$, can be generated by choosing the tag $1_c$ and $(m-1)$ distinct tags from set $\{2,\dots,N\}$.
	\item In Phase 2, the server may only serve the packet $1_c$ and the queue may contain the packet $1_n$ or not from the definition of Phase 2. Therefore, we define the service status $(H) \in \mathcal{H}_2$ as an $N$-permutation of set $\mathcal{I}_2 = \{1_c,2,\dots,N\}$ such that $h_N = 1_c$ which ensures that the packet under service belongs to the source-1. For any $(H) \in \mathcal{H}_1$, the term $C_m \in \Gamma_2$ for $1\leq m \leq N$ is defined as an $m$-combination of set $\{1_n,2,\dots,N\}$ whereas we use $\{C_m\} = \{0\}$ when the buffer is empty. 	
	\item For Phase 3, we define the service status $(H) \in \mathcal{H}_3$ as an $N$-permutation of set $\mathcal{I}_3 = \{1_n,2,\dots,N\}$. When the server is idle, we define the system state as $q = ((H),\{-1\})$ similar to the previously given definition for $Y(t)$. When the server is busy, the states are defined as follows: When the tag of packet under service is $1_n$, i.e., $h_N = 1_n$, the queue status is defined as $\{C_m\} = \{0\}$ since the packets behind $1_n$ are irrelevant in our model as discussed in the FSFS policy. When $h_N \neq 1_n$, the term $C_m \in \Gamma_3$ for $1\leq m \leq N$ is defined as an $m$-combination of set $\{1_n,2,\dots,N\}$ whereas we use $\{C_m\} = \{0\}$ when the buffer is empty. 	
	\item For Phase 4, we have a single state which we define as $q = (-1,(-1)) \in \mathcal{Q}_4$ similar to the FSFS policy.		
\end{itemize}
Finally, we list all the intra-phase and inter-phase transitions for the ESFS policy in Table~\ref{Table:intraPhase_ESFS} and Table~\ref{Table:interPhase_ESFS}, respectively. For intra-phase transitions, the rows 1-2, 3-4 and 5-9 refer to the transitions for Phase 1, 2 and 3, respectively. For inter-phase transitions, the rows 1, 2-3, 4, 5, and 6-7 refer to the transitions from Phase 1 to 2, Phase 2 to 3, Phase 3 to 4, Phase 4 to 2 and Phase 4 to 1, respectively. This concludes the analytical model for the ESFS policy.
\begin{table}[!htb]
	\centering
	\caption{Intra-phase transition rates $\alpha_{q,q'}$ for the ESFS policy}
%	\scalebox{1}{
		\begin{tabular}{|c|c|c|c|}
			\hline	
			$q$ 		   & $q'$                 						& $\alpha_{q,q'}$  & Condition 		\\ \hline
			
			%% Phase 1
			$((H),\{C_m\})$        & $((H),\{C_m, j\})$                           			& $\lambda_j$	  & $H \in \mathcal{H}_1, j \in \mathcal{J}_1,$
			\\ & & & $j\notin \{C_m\}, C_m \in \Gamma_1$	\\ \hline
			$((H),\{C_m\})$        & $((H'_{i^*}),\{C'_m\})$                          		& $\mu_{h_N}$				  & $H \in \mathcal{H}_1,$		 	
			\\ & & & $i^* \neq 1_c, C_m \in \Gamma_1$ \\ \hline
			
			%% Phase 2
			$((H),\{0\})$        & $((H),\{j\})$                          			& $\lambda_j$	  & $H \in \mathcal{H}_2, j \in \mathcal{J}_2$ \\ \hline
			$((H),\{C_m\})$        & $((H),\{C_m, j\})$                          			& $\lambda_j$	  & $H \in \mathcal{H}_2, j \in \mathcal{J}_2,$
			\\ & & & $j\notin \{C_m\}, C_m \in \Gamma_2$	\\ \hline
			
			%% Phase 3
			$((H),\{-1\})$        & $((H'_i),\{0\})$                          			& $\lambda_i$	  & $H \in \mathcal{H}_3, i \in \mathcal{I}_3$			 	\\ \hline
			$((H),\{0\})$        & $((H),\{j\})$                          			& $\lambda_j$	  & $H \in \mathcal{H}_3, h_N \neq 1_n,$			 	
			\\ & & & $j \in \mathcal{J}_3$ \\ \hline
			$((H),\{C_m\})$        & $((H),\{C_m, j\})$                          		& $\lambda_j$	  & $H \in \mathcal{H}_3, h_N \neq 1_n,$
			\\ & & & $j \in \mathcal{J}_3, j\notin \{C_m\},$		 	
			\\ & & & $C_m \in \Gamma_3$ \\ \hline
			$((H),\{0\})$        & $((H),\{-1\})$                          			& $\mu_{h_N}$	  &$H \in \mathcal{H}_3, h_N \neq 1_n$		 	\\ \hline
			$((H),\{C_m\})$        & $((H'_{i^*}),\{C'_m\})$                          			& $\mu_{h_N}$	  &$H \in \mathcal{H}_3, h_N \neq 1_n$			 	
			\\ & & & $C_m \in \Gamma_3$ \\ \hline
			
		\end{tabular}
%	}
	\label{Table:intraPhase_ESFS}
\end{table}
\begin{table}[!htb]
	\centering
	\caption{Inter-phase transition rates $\beta_{q,q'}$ for the ESFS policy}
%	\scalebox{0.915}{
		\begin{tabular}{|c|c|c|c|}
			\hline	
			$q$ 		   & $q'$                 						& $\beta_{q,q'}$  & Condition 		\\ \hline
			
			%% Phase 1 -> 2
			$((H),\{C_m\})$        & $((H'_{i^*}),\{C'_m\})$                          		& $\mu_{h_N}$				  & $H \in \mathcal{H}_1, i^* = 1_c$		 	
			\\ & & & $C_m \in \Gamma_1$ \\ \hline		
			
			%% Phase 2 -> 3
			$((H),\{0\})$        & $((H),\{-1\})$                          		& $\mu_1$				  & $H \in \mathcal{H}_2, h_N = 1_c $ \\ \hline
			$((H),\{C_m\})$        & $((H'_{i^*}),\{C'_m\})$                         			& $\mu_1$	  & $H \in \mathcal{H}_2, h_N = 1_c $ 
			\\ & & & $C_m \in \Gamma_2$ \\ \hline		
			%		$(1_c,(P_m))$        & $(s_1,(s_2,\dots,s_m))$                         			& $\mu_1$	  & $P_m \in \Gamma_2, m > 1$ \\ \hline
			
			%% Phase 3 -> 4
			$((H),\{0\})$        & $(-1,(-1))$                          		& $\mu_1$				  & $H \in \mathcal{H}_3, h_N = 1_n $ \\ \hline
			
			%% Phase 4 -> 2
			$(-1,(-1))$              & $((H'_i),\{0\})$                          		& $\pi_{((H),\{-1\})}$				  & $X(t) = 0, i = 1_c$\\ \hline
			
			%% Phase 4 -> 1
			$(-1,(-1))$              & $((H),\{1_c\})$                           		& $\pi_{((H),\{0\})} + $				  & $X(t) = 0, H \in \mathcal{H}_1$
			\\ & & $\pi_{((H),\{1\})}$ & \\ \hline
			
			$(-1,(-1))$              & $((H),\{C_m\})$                           		& $\pi_{((H),\{C_m\})} + $				  & $X(t) = 0, H \in \mathcal{H}_1$
			\\ & & $\pi_{((H),\{C_m,1\})}$ & $1_c \notin C_m, C_m \in \Gamma_1$ \\ \hline
			
		\end{tabular}
%	}
	\label{Table:interPhase_ESFS}
\end{table}
\subsection{Single Buffer With Replacement (SBR) Policy}
Each state for $Y(t)$ is enumerated as a tuple $q = (i,(j)) \in \mathcal{Q}_Y$ where  $i \in \mathcal{I}_Y = \{0,1,\dots,N\}$ enumerates the source tag of packet that is currently being served except the case when $i=0$ which is used for enumerating the idle server. For any $i > 0$, the element $(j)$ enumerates the buffer status such that $j \in \mathcal{B}_Y = \{0,1,\dots,N\}$ indicates the source tag of packet waiting in the buffer except the case when $j=0$ which is used for enumerating the empty buffer status. When the server is idle, i.e., $i = 0$, the only possible buffer status is $(j) = (0)$ since the server can be idle only when the buffer is empty. Thus, we enumerate the idle server status as $q = (0,(0))$ which completes the state definitions of $Y(t)$ for the SBR policy. Next, we provide the transitions rates $\nu_{q,q'}$ of $Y(t)$ in Table~\ref{Table:TransitionRates_SBR} where the rows 1-3 and 4-5 correspond to the arrival and departure events, respectively. Since the buffer is shared by all sources in the SBR policy, an arrival from any source replaces the existing packet in the buffer (in contrast with the FSFS and ESFS policies) as defined in row 3.
\begin{table}[!htb]
	\centering
	\caption{Transition rates $\nu_{q,q'}$ of $Y(t)$ for the SBR policy}
%	\scalebox{1}{
		\begin{tabular}{|c|c|c|c|}
			\hline
			$q$ 		   & $q'$                 						& $\nu_{q,q'}$  & Condition 		\\ \hline
			$(0,(0))$        & $(i,(0))$                          			& $\lambda_i$	  & $i \in \mathcal{N}$			 	\\ \hline
			$(i,(0))$        & $(i,(j))$                          			& $\lambda_j$	  & $i,j \in \mathcal{N}$			 	\\ \hline
			$(i,(j))$        & $(i,(k))$                          			& $\lambda_{k}$	  & $i,j,k \in \mathcal{N}, k \neq j$ 
			\\ \hline 
			$(i,(0))$        & $(0,(0))$                          			& $\mu_i$	  & $i \in \mathcal{N}$			 	\\ \hline
			$(i,(j))$        & $(j,(0))$                          			& $\mu_i$	  & $i,j \in  \mathcal{N}$			 	\\ \hline
			%		$(i,(P_m))$              & $(s_1,(s_2,\dots,s_m))$                          		& $\mu_i$				  & $i \in \mathcal{N}, P_m \in \Gamma_Y, m > 1$			 	\\ \hline
		\end{tabular}
%	}
	\label{Table:TransitionRates_SBR}
\end{table}
Next, we define the states $q = (i,(j)) \in \mathcal{Q}_p$ for each phase.

\begin{itemize}
	\item For Phase 1, the server status is defined as $i \in \mathcal{I}_1 = \{1_p,2,\dots,N\}$ similar to the FSFS policy. For any $i$ value, the buffer status can be only $(j) = (1_c)$ since the buffer always contains the packet $1_c$ in Phase 1.
	\item In Phase 2, the server may only serve the packet $1_c$, i.e., $i = 1_c$, and the buffer may contain the packet $1_n$ or not from the definition of Phase 2. Thus, we define the buffer status as $(j)$, $j \in \{0,1_n,2,\dots,N\}$, in Phase 2.
	\item For Phase 3, the server status is defined as $ i \in \mathcal{I}_3 = \{0, 1_n,2,\dots,N\}$ similar to the FSFS policy. When $i = 0$, the only possible buffer status is $(j) = (0)$ for which we have the idle server status. When $i = 1_n$, the only possible buffer status is also $(j) = (0)$ since any to-be-served packet after the packet $1_n$ is irrelevant in our model as discussed for the FSFS policy. For any other $i$ value, the buffer may be empty or hold a packet from any source for which we define the buffer status as $(j)$, $j \in \{0,1_n,\dots,N\}$, similar to Phase 2.
	\item For Phase 4, we have a single state which we define as $q = (-1,(-1)) \in \mathcal{Q}_4$ similar to the FSFS policy.	
\end{itemize}

Finally, we list all the intra-phase and inter-phase transitions for the SBR policy in Table~\ref{Table:intraPhase_SBR} and Table~\ref{Table:interPhase_SBR}, respectively. For intra-phase transitions, the rows 1-2 and 3-6 refer to the transitions for Phase 2 and 3, respectively. In contrast with the FSFS and ESFS policies, there is no intra-phase transition defined for Phase 1 since the first packet arrival from any source replaces the packet $1_c$ in the buffer which results in a direct transition to Phase 4 for the SBR policy. For inter-phase transitions, the rows 1, 2, 3-4, 5, 6 and 7-8 refer to the transitions from Phase 1 to 2, Phase 1 to 4, Phase 2 to 3, Phase 3 to 4, Phase 4 to 2 and Phase 4 to 1, respectively. The last row corresponds to the case where the packet $1_c$ finds the server busy upon its arrival to the system in which case it replaces the packet in the buffer irrespective of its source as opposed to to FSFS and ESFS policies. Thus, out of Phase 4, the system transitions into the state $q=(i,(1_c))$ with rate $\sum\limits_{j \in \mathcal{B}_Y} \pi_{(i,(j))}$ that is the sum of steady-state probabilities of all states in $Y(t)$ where the source-$i$ packet is being served. This concludes the analytical model for the SBR policy.
\begin{table}[!h]
	\centering
	\caption{Intra-phase transition rates $\alpha_{q,q'}$ for the SBR policy}
%	\scalebox{1}{
		\begin{tabular}{|c|c|c|c|}
			\hline	
			$q$ 		   & $q'$                 						& $\alpha_{q,q'}$  & Condition 		\\ \hline
			
			%% Phase 1
			%		$(i,(j))$        & $(i,(P_m,j))$                          			& $\lambda_j$	  & $i \in \mathcal{I}_1,$
			%		\\ & & & $j \in \mathcal{J}_1, j\notin \{P_m\},$			 	
			%		\\ & & & $P_m \in \Gamma_1, m < N$ \\ \hline
			%		$(i,(P_m))$              & $(s_1,(s_2,\dots,s_m))$                          		& $\mu_i$				  & $i \in \mathcal{I}_1,s_1 \neq 1_c$		 	
			%		\\ & & & $P_m \in \Gamma_1, m > 1$ \\ \hline
			
			%% Phase 2
			$(1_c,(0))$        & $(1_c,(j))$                          			& $\lambda_j$	  & $j \in \mathcal{J}_2$ \\ \hline
			$(1_c,(j))$        & $(1_c,(k))$                          			& $\lambda_k$	  & $j,k \in \mathcal{J}_2, k \neq j $ \\ \hline
			
			%% Phase 3
			$(0,(0))$        & $(i,(0))$                          			& $\lambda_i$	  & $i \in \mathcal{I}_3-\{0\}$			 	\\ \hline	
			$(i,(0))$        & $(i,(j))$                          			& $\lambda_j$	  & $i \in \mathcal{I}_3-\{0, 1_n\},j \in \mathcal{J}_3$			 	\\ \hline	
			
			$(i,(0))$        & $(0,(0))$                          			& $\mu_i$	  & $i \in \mathcal{I}_3-\{0, 1_n\}$		 	\\ \hline
			$(i,(j))$        & $(j,(0))$                          			& $\mu_i$	  & $i \in \mathcal{I}_3-\{0\}, j \in \mathcal{J}_3$			 	\\ \hline 
			%		$(i,(P_m))$              & $(s_1,(s_2,\dots,s_m))$                          		& $\mu_i$				  & $i \in \mathcal{I}_3-\{0, 1_n\},$		 	
			%		\\ & & & $P_m \in \Gamma_3, m > 1$ \\ \hline

		\end{tabular}
%	}
	\label{Table:intraPhase_SBR}
\end{table}
%\cellcolor{green}
\begin{table}[!h]
	\centering
	\caption{Inter-phase transition rates $\beta_{q,q'}$ for the SBR policy}
%	\scalebox{1}{
		\begin{tabular}{|c|c|c|c|}
			\hline	
			$q$ 		   & $q'$                 						& $\beta_{q,q'}$  & Condition 		\\ \hline
			
			%% Phase 1 -> 2
			$(i,(1_c))$              & $(1_c,(0))$                          		& $\mu_i$				  & $i \in \mathcal{I}_1$ \\ \hline
			
			%% Phase 1 -> 4
			$(i,(1_c))$              & $(-1,(-1))$                          		& $\lambda_j$				  & $i \in \mathcal{I}_1, j \in \mathcal{J}_1$ \\ \hline
			
			%% Phase 2 -> 3
			$(1_c,(0))$              & $(0,(0))$                          		& $\mu_1$				  & \\ \hline
			$(1_c,(j))$        & $(j,(0))$                          			& $\mu_1$	  & $j \in \mathcal{J}_2$ \\ \hline		
			
			%% Phase 3 -> 4
			$(1_n,(0))$              & $(-1,(-1))$                          		& $\mu_1$				  & \\ \hline
			
			%% Phase 4 -> 2
			$(-1,(-1))$              & $(1_c,(0))$                          		& $\pi_{(0,(0))}$				  & $X(t) = 0$\\ \hline
			
			%% Phase 4 -> 1
			%		$(-1,(-1))$              & $(i,(1_c))$                          		& $\pi_{(i,(P_m))}$				  & $X(t) = 0, i \in \mathcal{I}_1$
			%		\\ & & & $s_m \neq 1_c, P_m \in \Gamma_1$ \\ \hline
			
			$(-1,(-1))$              & $(i,(1_c))$                          		& $\sum\limits_{j \in \mathcal{B}_Y} \pi_{(i,(j))} $				  & $X(t) = 0, i \in \mathcal{I}_1,$ \\ \hline 
			
		\end{tabular}
%	}
	\label{Table:interPhase_SBR}
\end{table}

\subsection{Computational Considerations for the Analytical Models}
In this subsection, a comparison of the computational cost of the MFQ-based analytical model for each policy is provided. Note that, this is different than the complexity of implementing the actual policies on the server.
For this comparison, we report the size of the square matrix $\boldsymbol{W}_{L \times L}$ whose inversion is required for obtaining the average AoI for each source in \eqref{eq:moments} or the matrix exponential function of $\boldsymbol{W}$ is needed to obtain the age violation probabilities in \eqref{equation:finalPDF}. Clearly, we have $L = |\mathcal{Q}_1| + |\mathcal{Q}_2| + |\mathcal{Q}_3|$ which corresponds to the number of states in $\mathcal{S}$ with positive drift where the notation $|\cdot|$ is used to denote the cardinality of the argument set. The values of $L$ for each policy, denoted by $L_{SBR}$, $L_{FSFS}$, and $L_{ESFS}$, respectively, are listed in Table~\ref{Table:compComplexity} when the number of sources ranges between 2 and 5. We observe that as the number of sources increases, the size of the matrix $\boldsymbol{W}$ grows significantly faster with the ESFS and FSFS policies than the SBR policy which subsequently limits the number of sources that can be analyzed with the MFQ technique when the computational resources are limited. In fact, it is observed that the MATLAB implementation for the MFQ analysis given in Section~\ref{section:SystemModel} is feasible with personal computers when the number of sources is less than or equal to 5. When the number of sources increases further, further computational capabilities might be needed. However, we note that the proposed technique is computationally stable.
\begin{table}[!ht]
	\centering
	\caption{The values of $L_{SBR}$, $L_{FSFS},$ and $L_{ESFS}$, when the number of sources $N$ ranges between 2 and 5.}
	\begin{tabular}{|c|c|c|c|c|}
		\hline
		& $N$ = 2 & $N$ = 3 & $N$ = 4 & $N$ = 5 \\ \hline
		$L_{SBR}$  &   10    &   17    &    26    &   37      \\ \hline
		$L_{FSFS}$ &   16    &   65    &   326    &   1957    \\ \hline
		$L_{ESFS}$ &   15    &   80    &   606    &   5904    \\ \hline
	\end{tabular}
	\label{Table:compComplexity}
\end{table}

\section{Numerical Examples}
\label{section:NumericalExamples}
In this section, the proposed analytical model is first verified with simulations for each policy. Subsequently, the analytical model is used to compare the three studied policies under several scenarios.
\subsection{Model Validation}
We consider a scenario where $N$ = 4 for which the arrival and service rate vectors are assumed to be (1,2,3,2) and (3,1,2,4) packets/sec, respectively. The cumulative distribution function (CDF) of the AoI for each source-$n$, denoted by $F_{\Delta_n}(x)$, is shown in Fig.~\ref{figure:CDF} for each policy using both the analytical model and simulations. We observe that the analytical results are perfectly in line with the simulation results. Therefore, for the rest of the paper, we will only use the proposed analytical model for evaluating the policies.
\begin{figure*}[tbh]
	\centering
	\begin{subfigure}{.32\textwidth}
		\centering
		\includegraphics[width=\linewidth]{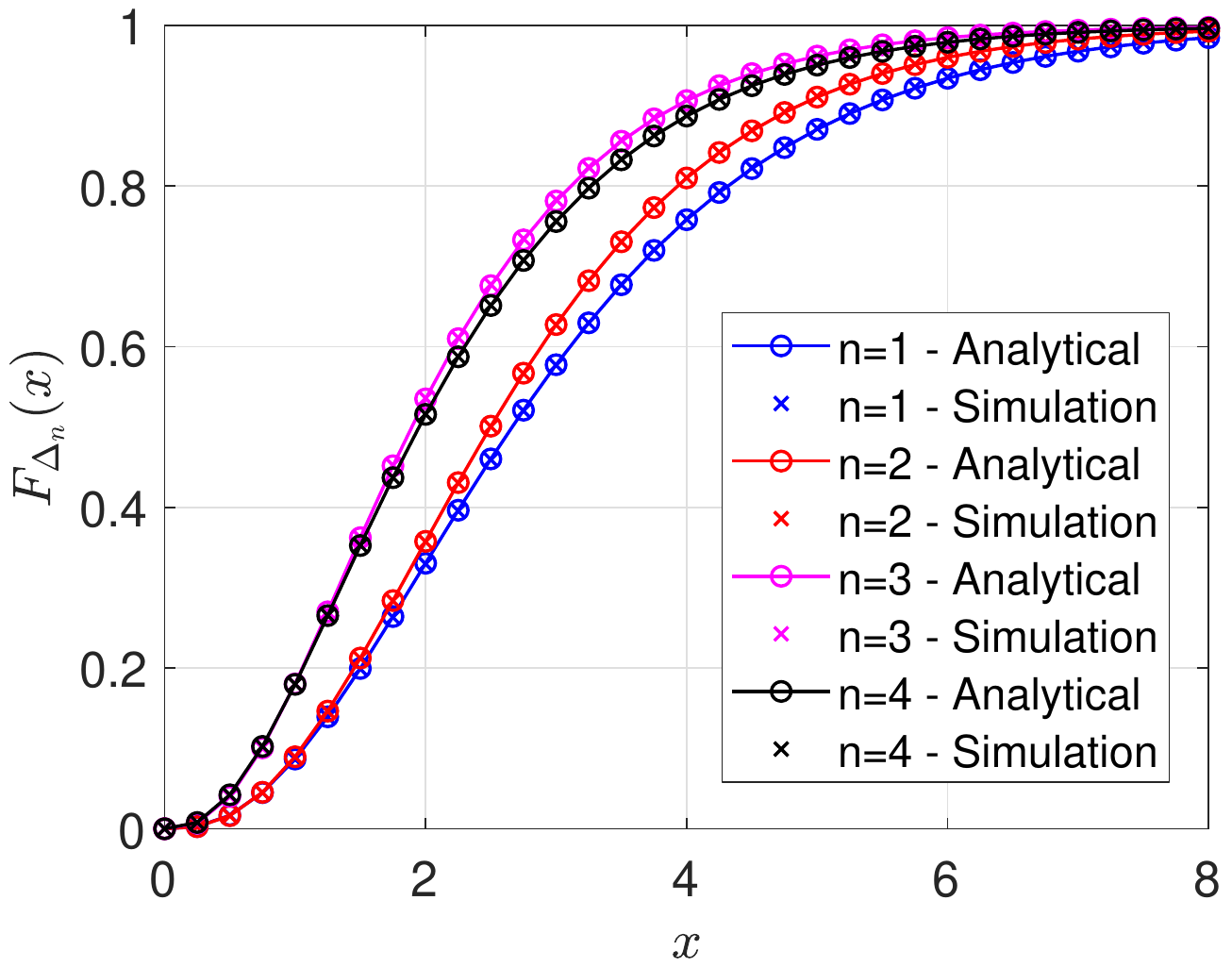}
		\caption{First Source First Serve (FSFS)}	
		\label{figure:sim_P_Tp}
	\end{subfigure}\hfill
	\begin{subfigure}{.32\textwidth}
		\centering
		\includegraphics[width=\linewidth]{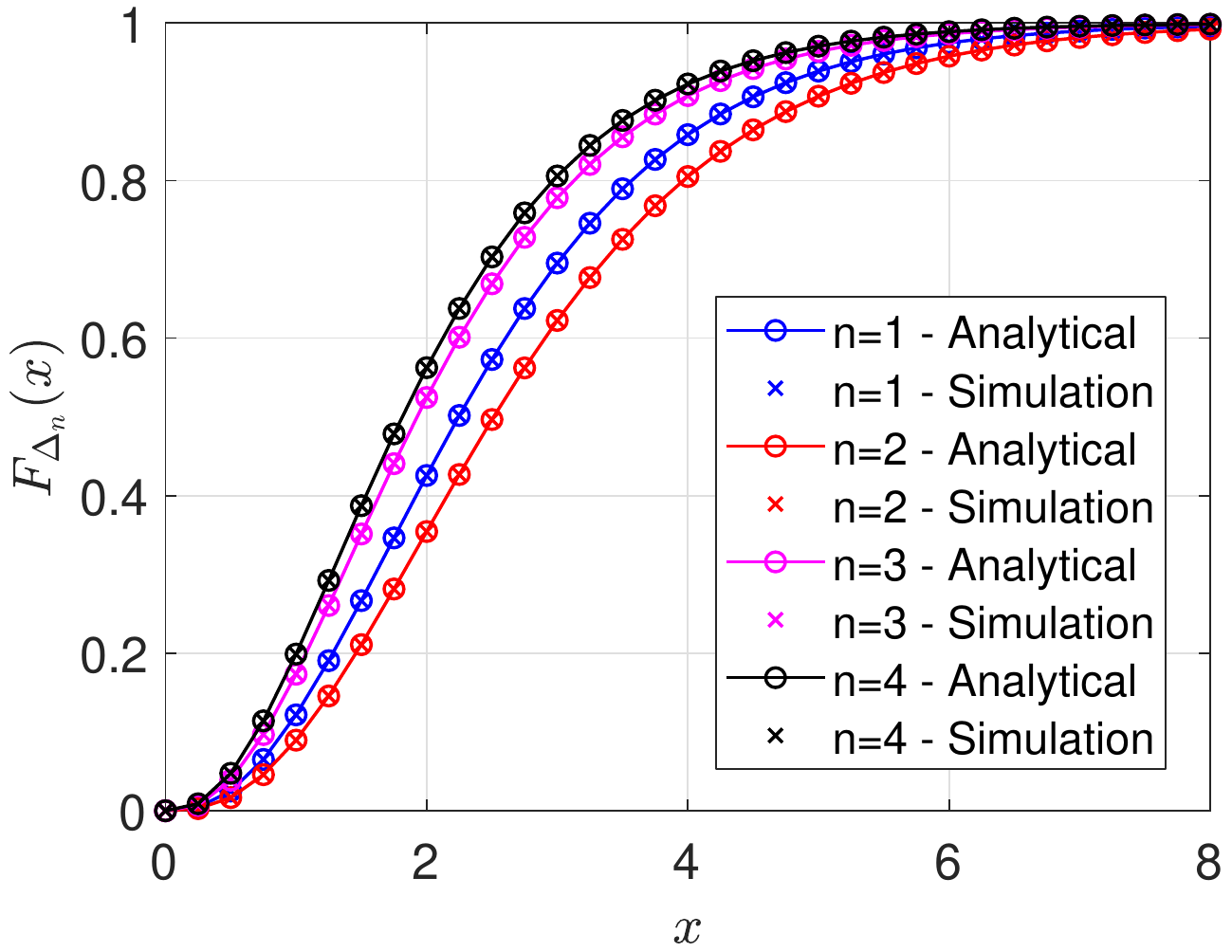}	
		\caption{Earliest Served First Serve (ESFS)}
	\end{subfigure}\hfill%
	\begin{subfigure}{.32\textwidth}
		\centering
		\includegraphics[width=\linewidth]{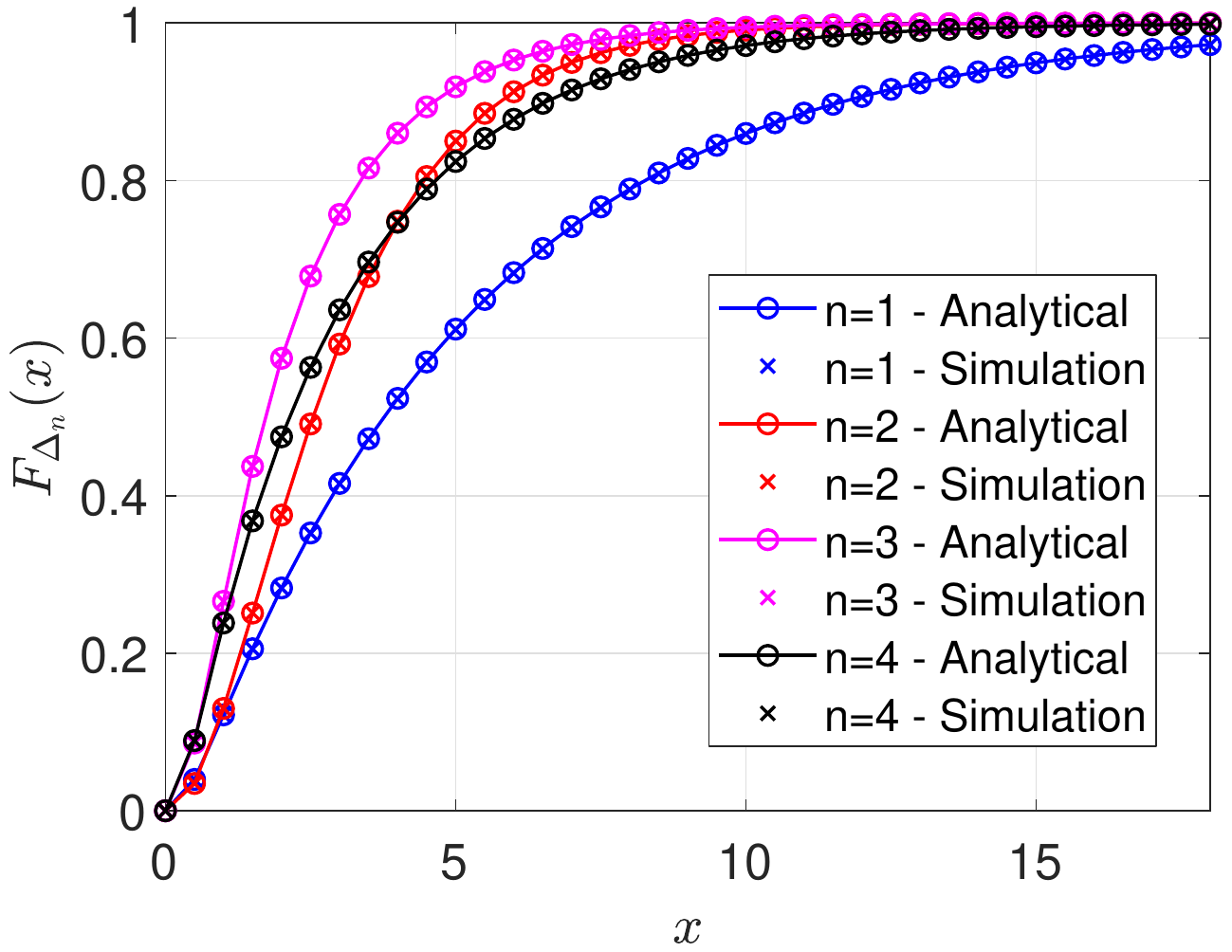}
		\caption{Single Buffer with Replacement (SBR)}
		\label{figure:sim_pv_Tp}
	\end{subfigure}
	\caption{The CDF $F_{\Delta_n}(x)$ of each source for the FSFS, ESFS, and SBR policies when the number of sources $N=4$, and the arrival and the service rate vectors are (1,2,3,2) and (3,1,2,4) packets/sec, respectively.}
	\label{figure:CDF}
\end{figure*}

\subsection{Comparative Assessment of the Scheduling Policies}
In this subsection, the performance of the studied policies are evaluated with respect to the average AoI and average age violation probability metrics under several scenarios where the sources may have identical or different traffic intensities, referred to as balanced and unbalanced load scenarios, respectively. We assume the service rate is common and equal to one for all sources, i.e., $\mu_n = 1, \forall n \in \mathcal{N}$, for all the numerical examples.

\subsubsection{Balanced Load}
In this subsection, we consider a scenario where the load is balanced among all sources such that the arrival rate for each source is given as $\lambda_n = \rho / N, \forall n \in \mathcal{N}$. We sweep the number of sources from 3 to 5 for which we obtain the average AoI for each policy with respect to the system load $\rho$ as shown in Fig.~\ref{figure:balancedFig}. We observe that the ESFS policy consistently outperforms the other two policies in moderate loads with FSFS being slightly worse for all the three cases. Moreover, the performance gaps between the policies grow as the number of sources increases. This shows the effectiveness of selecting the source that is not served for the longest duration as opposed to considering first packet arrival times of FSFS. Lastly, as the system load increases towards infinity, the average AoI for the ESFS and FSFS policies become identical as expected since both policies behave the same, i.e., round-robin service, when there is always a packet (in the waiting room) for each source upon a service completion.
\begin{figure}[tbh]
	\centering{
		\includegraphics[width=0.75\linewidth]{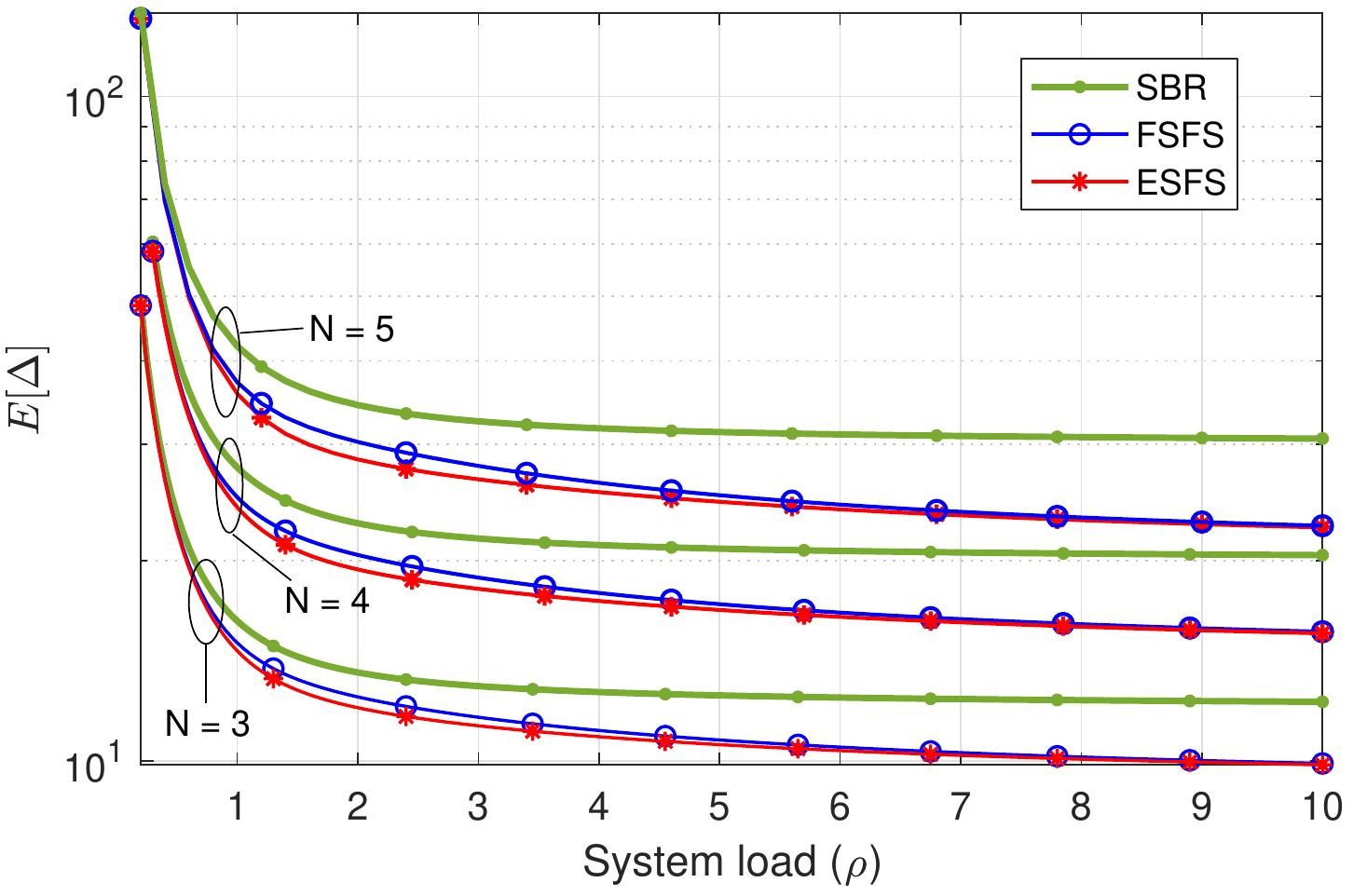}
	}
	\caption{The average age $E[\Delta]$ obtained with the SBR, FSFS, and ESFS policies as a function of the system load $\rho$ when all sources have same traffic intensities.}
	\label{figure:balancedFig}
\end{figure}

In the next example, we evaluate the studied policies with respect to the average age violation probability metric under two system loads. Specifically, the low and moderate load scenarios are considered where the parameter $\rho$ for each case is assumed to be 0.5 and 4, respectively. For both scenarios, the average age violation probability with respect to the age threshold parameter $\gamma$ is depicted in Fig.~\ref{figure:ageVio} for all three policies. We observe that when the system load is low, FSFS and ESFS policies perform quite close to each other with a slightly better performance than SBR policy whereas the performance gap grows in the moderate load. Moreover, the ESFS policy outperforms the other two policies in both scenarios. Lastly, as the system load increases, we observe that the average age violation probability can be reduced significantly faster with SBPSQ policies than it can be achieved with the SBR policy.

% f__vioAOI__4source_equalSourceLoad_combined, f__vioAOI__4source_equalSourceLoad__Grup
\begin{figure}[htb]
	\centering{
		\includegraphics[width=0.75\linewidth]{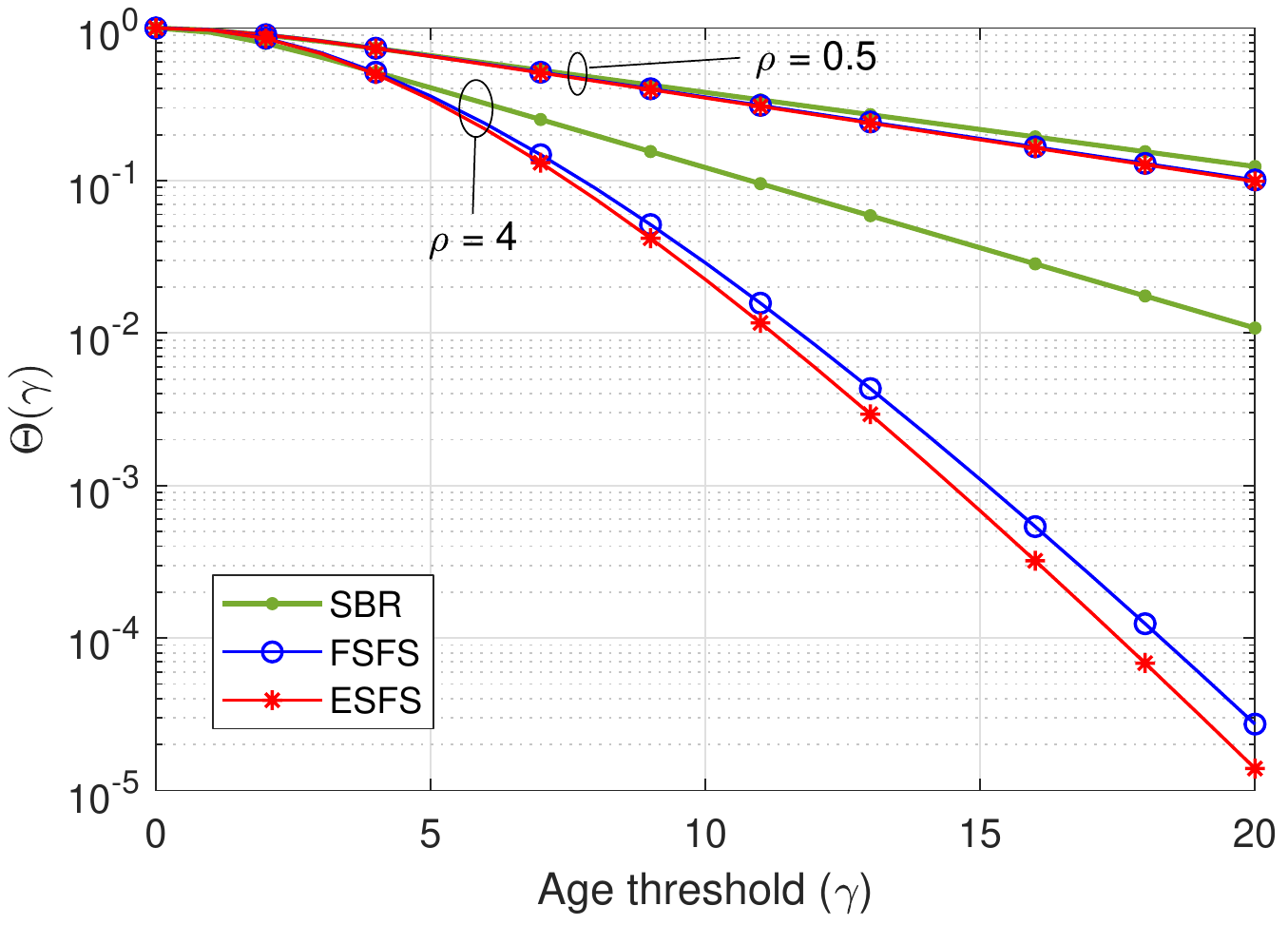}
	}
	\caption{The average age violation probability $\Theta(\gamma)$ for the SBR, FSFS, and ESFS policies as a function of the age threshold parameter $\gamma$ when there are $N$ = 4 sources with identical traffic intensities.}
	\label{figure:ageVio}
\end{figure}

\subsubsection{Unbalanced Load}
In this subsection, we study a scenario where the sources may have different traffic intensities given a fixed system load. We assume the number of sources $N$ = 2 for which the average AoI with respect to the source-1 load ratio, defined as $\rho_1 / \rho$, is given in Fig.~\ref{figure:unbalanced} for the low and moderate load scenarios (where we sweep $\rho_1$ from $\rho/2$ to $\rho$ due to symmetry). In the low load scenario, we observe that all three policies perform close to each other with SBR policy being slightly worse. In the moderate system load, the average AoI worsens with remarkably slower rate for SBPSQ policies than SBR policy as the load asymmetry between the sources increases. Morever, we observe that the ESFS policy consistently outperforms FSFS and SBR policies for any $\rho_1$ value under both system loads. This shows the effectiveness of the ESFS policy also under  scenarios with different traffic mixes.

\begin{figure}[htb]
	\centering{
		\includegraphics[width=0.75\linewidth]{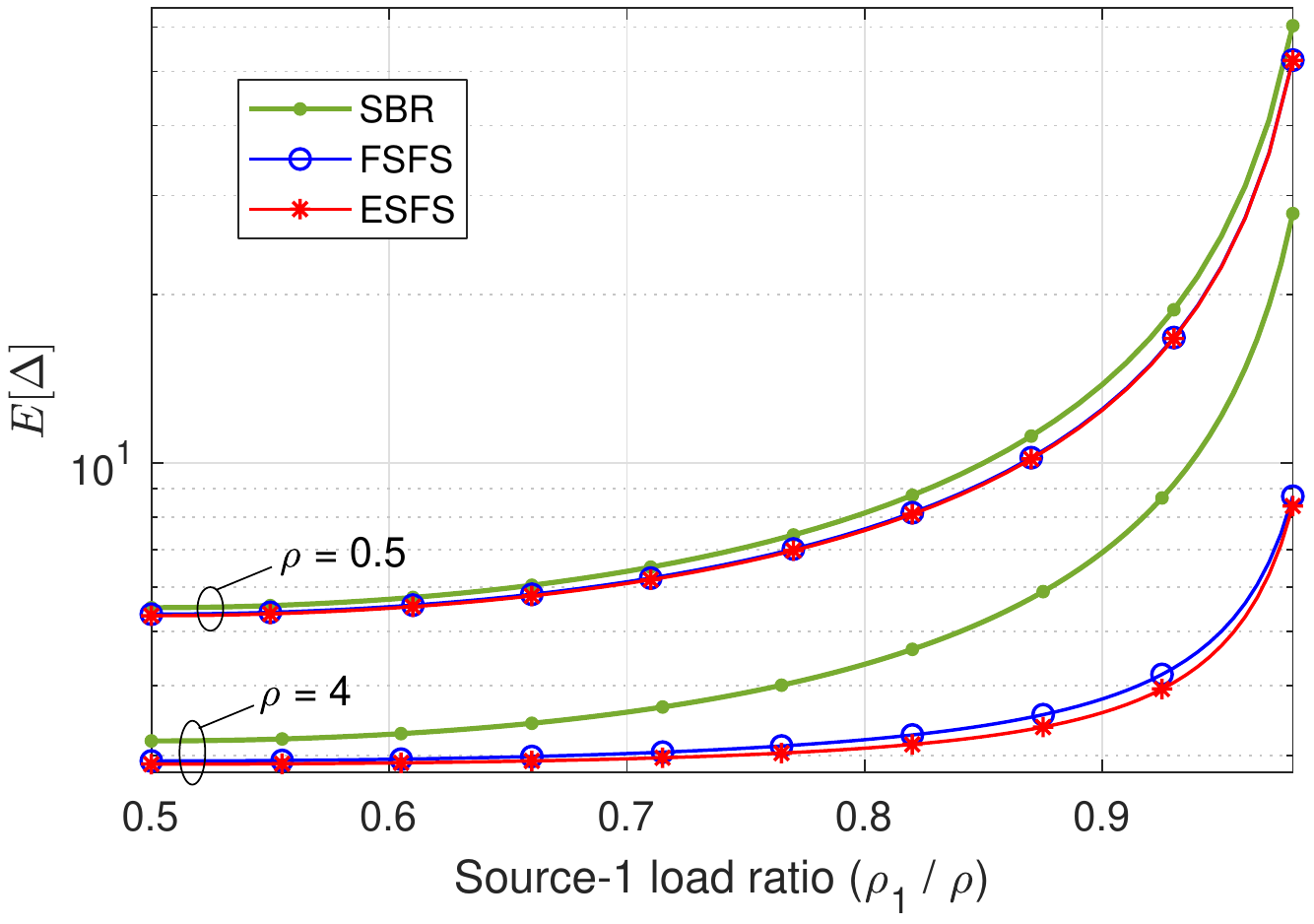}
	}
	\caption{The average age $E[\Delta]$ obtained with the SBR, FSFS and ESFS policies as a function of the source-1 load ratio ($\rho_1 / \rho$) when the number of sources $N$ = 2.}
	\label{figure:unbalanced}
\end{figure}

\section{Conclusions}
\label{section:Conclusions}
In this paper, we study a multi-source information update system where the sources send status updates to a remote monitor through a single server. Under the assumption of Poisson packet arrivals and exponentially distributed heterogeneous service times for each source, we propose and validate an analytical model to obtain the exact steady-state distributions of the AoI process for each source under several queueing policies. The average AoI and the average age violation probabilities are then easily calculated from the obtained distributions which are in matrix exponential form. In the numerical examples, we evaluated the studied policies for several scenarios under a common service time distribution with varying system loads and different traffic mixes. We show that the proposed ESFS policy which is age-agnostic and simple-to-implement, consistently outperforms the other two studied policies where the degree of outperformance with respect to FSFS being modest. Furthermore, when SBPSQ policies are employed at the server, the performance improvement with SBPSQ policies over SBR increases with higher loads and also when the load asymmetry among the sources increases. Future work will consist of practical scheduling policies for non-symmetric networks with heterogeneous service times when the minimization of weighted average AoI is sought. 

%\bibliographystyle{unsrtnrt}

%\bibliographystyle{unsrtnat}
  %%% Uncomment this line and comment out the ``thebibliography'' section below to use the external .bib file (using bibtex) .

\end{document}